\def\PUBLICVERSION{1}
\documentclass{article}
\usepackage{iclr2027_conference,times}
\usepackage[utf8]{inputenc}
\usepackage[T1]{fontenc}
\usepackage{booktabs}
\usepackage{amsfonts,amsmath,amssymb,bm,mathtools,amsthm}
\usepackage{nicefrac}
\usepackage{microtype}
\usepackage[table]{xcolor}
\usepackage{graphicx}
\usepackage{multirow}
\usepackage{adjustbox}
\usepackage{wrapfig}
\usepackage{xspace}
\usepackage[shortlabels]{enumitem}
\usepackage{subcaption}
\usepackage{float}
\usepackage{needspace}
\usepackage{array}
\usepackage{tabularx}
\usepackage{longtable}
\usepackage{makecell}
\usepackage{fontawesome5}
\usepackage{tikz}
\usetikzlibrary{shapes.geometric, arrows.meta, positioning, fit, calc}
\usepackage{url}
\usepackage{hyperref}

\definecolor{revisionred}{RGB}{180, 0, 0}
\definecolor{bestbg}{RGB}{221, 238, 221}
\definecolor{secondbg}{RGB}{255, 242, 204}
\definecolor{thirdbg}{RGB}{245, 245, 245}
\definecolor{linkbrown}{rgb}{0.55,0.30,0.10}
\definecolor{cvprblue}{rgb}{0.31, 0.59, 0.79}

\newif\ifredline
\redlinefalse

\newcommand{\best}[1]{\cellcolor{bestbg}\textbf{#1}}

\newcommand{\method}{\textsc{Doctor Trigger}\xspace}

\newcommand{\dfh}{image-only detector\xspace}

\newtheorem{proposition}{Proposition}

\definecolor{darkgreen}{rgb}{0.0, 0.5, 0.0}
\hypersetup{
  colorlinks=true,
  linkcolor=cvprblue,
  citecolor=darkgreen,
  urlcolor=blue,
  pdftitle={Doctor Trigger: A Framework for Release-Bound Face Verification},
  pdfauthor={\ifdefined\PUBLICVERSION Chia-Ming Lee, Chia-Yu Lin, Hung-Kai Huang, Yi-Ting Ku, Yu-Chen Liang, Chih-Chung Hsu\else Anonymous Author(s)\fi},
  pdfsubject={Release-bound verification of enrolled face assets},
  pdfkeywords={content credentials, media provenance, face verification, Fourier phase, enrollment records},
}

\title{\textsc{Doctor Trigger}: A Framework for \\Release-Bound Face  \mbox{Verification}}

\ifdefined\PUBLICVERSION
\iclrfinalcopy
\author{%
\textbf{Chia-Ming Lee$^{1,2}$ \quad Chia-Yu Lin$^{3}$ \quad Hung-Kai Huang$^{2}$} \\
\textbf{Yi-Ting Ku$^{3}$ \quad Yu-Chen Liang$^{3}$ \quad Chih-Chung Hsu$^{2}$} \\[3pt]
\normalfont $^{1}$University at Albany, SUNY \quad
\normalfont $^{2}$National Yang Ming Chiao Tung University \\
\normalfont $^{3}$National Cheng Kung University
\vspace{-4mm}
}
\else
\author{Anonymous Author(s)\\
Affiliation\\
\texttt{email}}
\fi

\begin{document}

\maketitle
\ifdefined\PUBLICVERSION\lhead{Preprint}\fi

\vspace{-3mm}
\begin{abstract}
\vspace{-2mm}
Face swapping and generative editing enable realistic facial manipulation, while online redistribution complicates publisher oversight of circulating face images. Existing approaches establish provenance, detect manipulation, recover embedded watermarks, verify image integrity, or disrupt unauthorized face manipulation. These capabilities do not necessarily establish consistency with a specific publisher-enrolled release, limiting the evidence available for verification and human review. To address this gap, we introduce \method, a framework for verifying circulating faces against private publisher-held release records. Before publication, the publisher embeds a weak keyed Fourier-phase signal in each face while preserving visual fidelity and stores its spectral reference. Afterward, verification compares the circulating face with that reference and uses key projection to measure retention of the release-specific signal. Therefore, relation and key evidence jointly help publishers identify suspicious copies for review, potentially reducing the circulation of manipulated images. Experimental results across FF++, DFDC, and Celeb-DF show that spectral agreement separates tested face swaps and generative manipulations from enrolled releases, while the keyed signal distinguishes published versions from their unreleased originals under the evaluated key-preserving benign transformations. \method complements provenance metadata with release-specific face verification and record-linked evidence for further assessment.
\end{abstract}

\vspace{-5mm}\section{Introduction}
\label{sec:intro}

Recent face and head swapping systems produce realistic identity changes~\citep{Wang2025DynamicFace,Kang2025HeadSwapping}, while compression and resizing also alter images as they circulate through online sharing and reposting. For publishers such as news organizations and public institutions, the question is: \emph{does a publicly distributed face still correspond to the specific version that was released?} The publisher can retain release-specific evidence before distribution for later verification.

To support trust in circulating images, C2PA binds signed provenance assertions to content and supports reconnection through soft bindings~\citep{c2pa}. Complementing provenance, recent detectors learn visual-language cues, discrepancies, and large-scale forgery representations~\citep{Guo2025M2F2Det,Yang2025D3,Jung2026SD2}. Reference-based detection further compares regional identity and texture with genuine references~\citep{Ming2026Reference}. However, these capabilities may not establish whether a circulating face remains consistent with a specific publisher-enrolled release.

To address this need, \method uses private records for \textbf{release-bound face verification}: assessing circulating faces against specific publisher-enrolled releases. First, the publisher applies a weak keyed phase displacement before publication and stores a spectral reference. A publisher-operated verification service, such as a newsroom verification desk, then uses the record to assess face consistency and release-signal retention in user-submitted copies. Finally, the resulting status guides routine handling or further assessment, as illustrated in Fig.~\ref{fig:positioning}. The name reflects this workflow: \emph{Doctor} evokes assessment against a retained history followed by triage; \emph{Trigger} denotes the weak phase signal added before publication while preserving visual fidelity.

\begin{figure}[t]
\centering
\vspace{-3mm}
\resizebox{0.98\textwidth}{!}{%
\begin{tikzpicture}[
  font=\normalsize,
  ent/.style={rectangle, rounded corners=3pt, draw=blue!25!gray, thick, fill=blue!3, minimum height=0.95cm, align=center, inner sep=4pt},
  priv/.style={rectangle, rounded corners=3pt, draw=orange!65!brown, thick, fill=orange!7, minimum height=0.95cm, align=center, inner sep=4pt},
  pub/.style={rectangle, rounded corners=3pt, draw=teal!65!black, thick, fill=teal!5, minimum height=0.95cm, align=center, inner sep=4pt},
  flow/.style={-{Stealth[length=6pt]}, line width=1.05pt, gray!85!black},
  privflow/.style={-{Stealth[length=6pt]}, line width=1.5pt, orange!80!brown, dash pattern=on 7pt off 4pt},
  pubflow/.style={-{Stealth[length=6pt]}, line width=1.8pt, teal!65!black, line cap=round, dash pattern=on 0pt off 5pt},
]
\path[draw=gray!45, fill=gray!3, line width=0.8pt, rounded corners=5pt]
 (-1.5,1.45) -- (10.35,1.45) -- (10.35,-1.45)
 -- (6.3,-1.45) -- (6.3,-3.45) -- (-1.5,-3.45) -- cycle;
\path[draw=blue!25!gray, fill=blue!2, line width=0.8pt, rounded corners=5pt]
 (10.65,1.45) -- (18.5,1.45) -- (18.5,-6.05)
 -- (1.0,-6.05) -- (1.0,-3.65) -- (6.6,-3.65)
 -- (6.6,-1.65) -- (10.65,-1.65) -- cycle;
\node[anchor=west, font=\normalsize\bfseries, text=gray!75!black]
 at (-1.25,1.15) {1. Enrollment \& release};
\node[anchor=west, font=\normalsize\bfseries, text=blue!35!black]
 at (10.9,1.15) {2. Verification \& handling};
\node[ent, text width=2.2cm]  (up)    at (0,0)     {\faUserEdit\ Original\\uploader};
\node[ent, text width=2.7cm]  (plat)  at (3.7,0)   {\faServer\ Publisher $P$};
\node[pub, text width=3.1cm]  (rel)   at (8.0,0)   {\faImage\ Released image $\mathbf{I}'$\\ \small \faStamp\ proposed credential};
\node[ent, text width=2.6cm]  (query) at (12.3,0)  {\faRetweet\ Query $\mathbf{J}$\\ \small circulating copy};
\node[ent, text width=2.8cm]  (user)  at (16.7,0)  {\faUsers\ Requester $U$};
\node[priv, text width=3.3cm] (store) at (3.7,-2.7) {\faLock\ Private record $R_{a,i}$\\ \small key, thresholds};
\node[priv, draw=orange!80!brown, line width=1.2pt, rounded corners=5pt,
  minimum height=1.30cm, text width=4.1cm] (ver) at (9.6,-2.7)
  {\textcolor{blue!35!black}{\bfseries\textsc{Doctor Trigger}}\\[-1pt]
   \small private verifier $V$\\[2pt]
   \small face consistency\quad release signal};
\coordinate (emblem) at ($(ver.north west)+(0.14,-0.03)$);
\fill[white] (emblem) circle (0.40);
\draw[orange!80!brown, line width=1.0pt] (emblem) circle (0.38);
\node[inner sep=0pt, text=blue!35!black] at (emblem)
  {\fontsize{17}{17}\selectfont\faUserMd};
\node[pub, text width=4.0cm]  (stat)  at (16.1,-2.7) {\faStamp\ Verification receipt\\ \small status \& coverage};

\draw[flow] (up) -- node[above=3pt, font=\small, align=center] {selects\\faces} (plat);
\draw[flow] (plat) -- node[above=3pt, font=\small] {releases} (rel);
\draw[flow] (rel) --
  node[below=15pt, font=\small, align=center, fill=white, fill opacity=0.7, text opacity=1, inner sep=1pt]
  {\textcolor{green!35!black}{benign processing}\\[-2pt] {\scriptsize\textcolor{green!35!black}{(JPEG, resize, compression)}}\\[-1pt] or \textcolor{red!65!black}{face swap / reenactment}\\[-2pt] \textcolor{red!65!black}{/ generative editing}}
  (query);
\draw[pubflow, rounded corners=3pt] (rel.north) -- (8.0,0.85)
  -- (16.7,0.85) -- (user.north);
\node[above=2pt, inner sep=1pt, font=\small, text=teal!65!black]
  at (8.6,0.85) {release credential};
\draw[flow] (user.west) -- node[below=4pt, font=\small] {submits} (query.east);
\draw[flow, rounded corners=3pt] ([xshift=6mm]query.south) -- (12.9,-2.38)
  -- ([yshift=3.2mm]ver.east);
\draw[privflow] (plat) -- node[right=4pt, font=\small, align=left, text=orange!80!brown] {enrolls\\each face} (store);
\draw[privflow] (store) -- (ver);
\draw[flow] (ver) -- (stat);
\draw[pubflow] ([xshift=6mm]stat.north) -- node[left=4pt, font=\small, align=right, text=teal!65!black] {verification\\result} (user.south);

\node[pub, draw=green!45!black, fill=green!13, line width=1.2pt, text width=3.7cm, minimum height=0.75cm] (confirmed) at (3.7,-5.5)
 {\textcolor{green!35!black}{\faUser\ \textbf{Confirmed}}};
\node[ent, draw=red!70!black, fill=red!10, line width=1.2pt, text width=3.7cm, minimum height=0.75cm] (unconfirmed) at (9.6,-5.5)
 {\textcolor{red!65!black}{\faUser\ \textbf{Unconfirmed}}};
\node[ent, draw=blue!45!gray, fill=blue!12, line width=1.2pt, text width=3.7cm, minimum height=0.75cm] (unavailable) at (16.1,-5.5)
 {\textcolor{blue!40!black}{\faUser\ \textbf{Unavailable}}};
\draw[gray!85!black,line width=1.05pt] (stat.south)
  -- node[right=4pt, font=\small] {decision} (16.1,-3.85) -- (1.35,-3.85);
\draw[flow] (1.35,-3.85) |- (confirmed.west);
\draw[flow] (7.25,-3.85) |- (unconfirmed.west);
\draw[flow] (13.75,-3.85) |- (unavailable.west);
\node[font=\small, align=left, anchor=west, text=green!35!black] at (1.65,-4.65)
 {Release consistency verified;\\retain separate review flags};
\node[font=\small, align=left, anchor=west, text=red!65!black] at (7.55,-4.65)
 {Flag for human review;\\withhold confirmation};
\node[font=\small, align=left, anchor=west, text=blue!40!black] at (14.05,-4.65)
 {Request inputs \& retry;\\if unresolved, human review};

\draw[gray!85!black,line width=1.05pt] (-1.25,-1.25)--(-0.85,-1.25);
\node[anchor=west, align=left, font=\small, text=gray!60!black, inner sep=0pt]
 at (-0.65,-1.25) {workflow};
\draw[orange!80!brown,line width=1.5pt,dash pattern=on 7pt off 4pt]
 (-1.25,-2.10)--(-0.85,-2.10);
\node[anchor=west, align=left, font=\small, text=gray!60!black, inner sep=0pt]
 at (-0.65,-2.10) {private\\publisher only};
\draw[teal!65!black,line width=1.8pt,line cap=round,dash pattern=on 0pt off 5pt]
 (-1.25,-2.95)--(-0.85,-2.95);
\node[anchor=west, align=left, font=\small, text=gray!60!black, inner sep=0pt]
 at (-0.65,-2.95) {public\\signed layer};
\end{tikzpicture}
}
\vspace{-2mm}
\caption{\textbf{Two-stage publisher-side verification interface.} The original uploader provides enrollment content; a requester, who may be a different person or platform, submits a circulating query and receives the verification result. The release-to-query arrow traces processing during circulation. Private records remain with the publisher (orange); credentials and status form the public layer (teal). The lower row proposes handling for each status; confirmation can coexist with a review flag. C2PA integration and public signing are proposed; retention tests use HMAC-authenticated credentials.}\vspace{-3mm}
\label{fig:positioning}
\label{fig:overall_pipeline}
\label{fig:overview}
\end{figure}
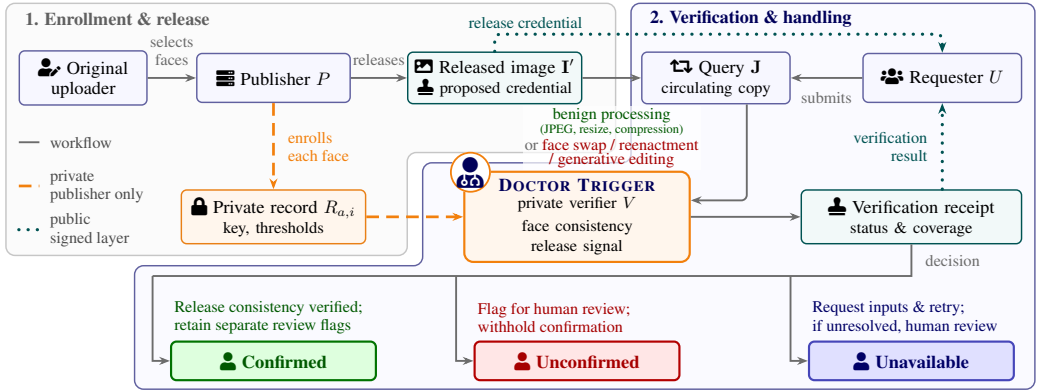

Conceptually, the reference assesses face consistency; the key helps identify the published version. Although original-phase matching, released-image similarity, and ArcFace features~\citep{Deng2019ArcFace} separate many clean face re-renderings, these comparisons may not distinguish the published version from its visually similar original under thresholds calibrated for benign processing. Key projection adds evidence of the signal introduced during enrollment, helping identify the enrolled release and providing \textbf{release binding}.

From a complementary perspective, robust watermarking preserves embedded messages through processing~\citep{Sander2025WAM,Alam2025SpecGuard,Bui2025TrustMark}, while semi-fragile watermarking aims to tolerate benign transformations while detecting manipulation~\citep{Neekhara2024FaceSigns}. However, message survival alone does not establish face consistency. Beyond message recovery, face-specific methods provide further evidence: RecoverMark supports localization and recovery~\citep{An2026RecoverMark}. Recently, LIDMark has combined landmarks and source identifiers for detection, localization, and tracing~\citep{Wu2026LIDMark}. Our private record supports separate assessments of face consistency and release-signal retention without requiring access to the original image at verification.

Building on these complementary roles, \method supports a \textbf{human-in-the-loop workflow} through its verification statuses: \emph{unconfirmed} directs reviewers to record discrepancies, while \emph{unavailable} requests missing information or human assistance. These outcomes \emph{trigger} the human \emph{doctor} to assess the case. Human-review evaluation was conducted under a confidentiality agreement; details and results cannot be disclosed in this paper at the participating party's request. We contribute:
\begin{itemize}[leftmargin=0pt]
    \item We formulate \textbf{release-bound face verification} using a private per-release record, complementing public provenance metadata and query-only forensic detection.

    \item We introduce a weak keyed phase displacement with complementary evidence: spectral agreement assesses face consistency to distinguish evaluated swaps and generative manipulations, while signed first-order projection measures release-signal retention, separating correct from wrong keys and releases from unreleased originals under evaluated benign processing.

    \item We evaluate release verification across datasets, isolating the contribution of key evidence through private-reference controls and assessing practical cost through deployment measurements. To further guide deployment, we analyze processing and attack risks, motivating channel validation and stronger binding between release signals and facial content.
\end{itemize}

\vspace{-3mm}
\section{Problem Setup}
\label{sec:setup}\vspace{-1mm}
\noindent\textbf{Enrollment and query.}
The publisher verifies a circulating face against a specific release enrolled before publication. Let $\mathbf I\in\mathbb R^{N_h\times N_w\times3}$ denote the original RGB face crop, with height $N_h$ and width $N_w$ in pixels. Enrollment adds a weak keyed signal to produce the released crop $\mathbf I'$ and retains a private record $R$. For asset $a$ and face slot $i$, this record is $R_{a,i}$; we omit the indices when considering one face. At verification, a request supplies the circulating query $\mathbf J$ and an identifier resolving its record. Our experiments verify individual crops with known enrollment geometry; the proposed multi-face interface verifies each associated face slot separately.

\vspace{-1mm}
\noindent\textbf{Processing channel.}
Once released, a face may undergo delivery processing. A benign channel $\mathcal P$ specifies recompression, resizing, or another such operation; $\mathcal C$ contains allowed, calibrated channels. Let $\mathcal R$ contain the evaluated face re-rendering operators $M$, such as swapping and diffusion resynthesis~\citep{Chen2020SimSwap,Rombach2022StableDiffusion}. We compare two processing histories:
\begin{equation}
\mathcal H_c:\ \mathbf J=\mathcal P(\mathbf I'),\qquad
\mathcal H_d:\ \mathbf J=\mathcal P(M(\mathbf I')),\quad
\mathcal P\in\mathcal C,\ M\in\mathcal R.
\label{eq:hypotheses}
\end{equation}
Here $\mathcal H_c$ represents delivery processing alone, while $\mathcal H_d$ includes face re-rendering before delivery; neither specifies the editor's intent. A supported channel must therefore preserve useful separation between these histories. Benign queries calibrate thresholds, while manipulation tests assess detection performance. The evaluated verifier receives the channel and crop geometry as inputs rather than inferring an unknown processing history.

\noindent\textbf{Verification and human assessment.}
Verification assesses face consistency and release-signal retention to support publisher decisions. In the proposed workflow, routine cases proceed automatically, while record discrepancies or incomplete verification prompt human assessment. Reviewers use the available scores and record information to examine discrepancies, request missing inputs, and decide whether further verification is needed before acting on a circulating copy. The decision policy formalizes this handoff through verification statuses and a separate review flag.
\vspace{-2mm}
\section{The Doctor Trigger Framework}
\label{sec:method}\vspace{-2mm}
\method prepares each face for verification before publication and assesses circulating copies against the retained release record. We first describe how the publisher creates the released image and its record, then introduce the verification scores and the policy that maps them to confirmed, unconfirmed, or unavailable outcomes.

\begin{figure}[!htbp]
    \centering\vspace{-2mm}
    \raisebox{1.8mm}{\includegraphics[width=0.48\columnwidth]{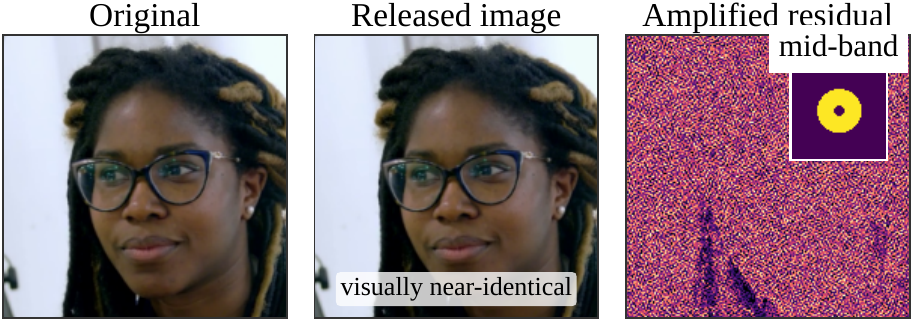}}
    \hfill
    \includegraphics[width=0.50\columnwidth]{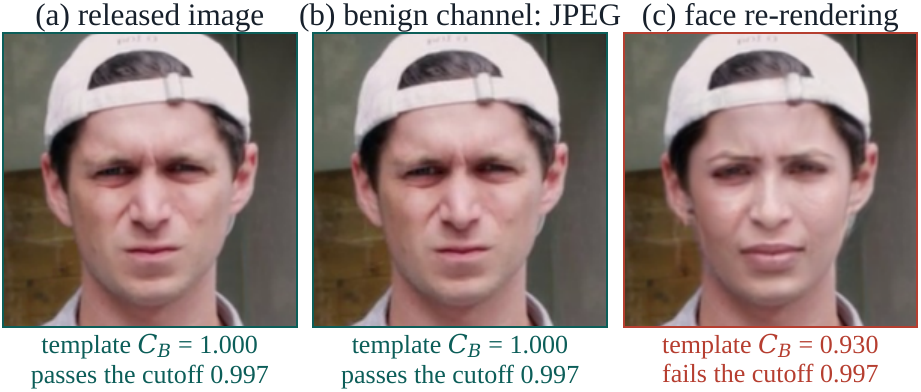}\vspace{-2mm}
\caption{\textbf{Release evidence with preserved visual fidelity.} Left: a weak keyed phase displacement adds a release-specific signal while preserving facial appearance. Right: template agreement $C_B$ illustrates why retaining a spectral reference enables face comparison after JPEG processing and reveals changes caused by re-rendering.}
\vspace{-2mm}
\label{fig:trigger_overview}
\end{figure}

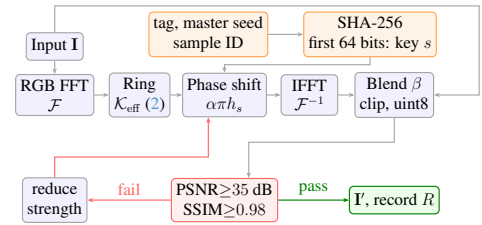
\begin{wrapfigure}{r}{0.45\columnwidth}
\vspace{-6mm}
\centering
\resizebox{0.44\columnwidth}{!}{%
\begin{tikzpicture}[
    node distance=0.45cm and 0.35cm,
    box/.style={rectangle, rounded corners=3pt, draw=gray!70, thick, fill=blue!6,
      minimum height=0.68cm, minimum width=1.4cm, align=center, font=\large},
    keybox/.style={rectangle, rounded corners=3pt, draw=orange!70, thick, fill=orange!10,
      minimum height=0.62cm, minimum width=1.2cm, align=center, font=\large},
    gatebox/.style={rectangle, rounded corners=3pt, draw=red!60, thick, fill=red!6,
      minimum height=0.68cm, minimum width=1.5cm, align=center, font=\large},
    outbox/.style={rectangle, rounded corners=3pt, draw=green!50!black, thick, fill=green!8,
      minimum height=0.68cm, minimum width=1.3cm, align=center, font=\large},
    arrow/.style={-{Stealth[length=4pt]}, thick, gray!70},
    redarrow/.style={-{Stealth[length=4pt]}, thick, red!60},
    greenarrow/.style={-{Stealth[length=4pt]}, thick, green!50!black},
]
\node[box] (fft)     at (0,    0)    {RGB FFT\\$\mathcal{F}$};
\node[box] (annulus) at (2.1,  0)    {Ring\\$\mathcal{K}_{\rm eff}$~\eqref{eq:annulus}};
\node[box] (phase)   at (4.2,  0)    {Phase shift\\$\alpha\pi h_s$};
\node[box] (ifft)    at (6.3,  0)    {IFFT\\$\mathcal{F}^{-1}$};
\node[box] (blend)   at (8.4,  0)    {Blend $\beta$\\clip, uint8};
\node[box] (input)   at (0,    1.2)  {Input $\mathbf{I}$};
\node[keybox] (prg) at (7.8,1.5) {SHA-256\\first 64 bits: key $s$};
\node[keybox] (seed) at (3.8,1.5) {tag, master seed\\sample ID};
\draw[arrow] (seed.east) -- (prg.west);
\node[gatebox] (gate) at (4.2, -2.6)
 {PSNR${\ge}35$ dB\\$\text{SSIM}{\ge}0.98$};
\node[box] (halve)   at (0, -2.6)  {reduce\\strength};
\node[outbox] (output) at (8.4, -2.6) {$\mathbf{I}'$, record $R$};
\draw[arrow] (input.south) -- (fft.north);
\draw[arrow] (input.north) -- +(0,0.65) -| (10.5,0) -- (blend.east);
\draw[arrow] (fft)     -- (annulus);
\draw[arrow] (annulus) -- (phase);
\draw[arrow] (phase)   -- (ifft);
\draw[arrow] (ifft)    -- (blend);
\draw[arrow] (prg.south) -- (7.8,0.75) -- (4.2,0.75) -- (phase.north);
\draw[arrow] (blend.south) -- +(0,-0.6) -| ([xshift=6mm]gate.north);
\draw[greenarrow] (gate.east) -- node[above, font=\large, text=green!50!black]{pass} (output.west);
\draw[redarrow] (gate.west) -- node[above, font=\large, text=red!60]{fail} (halve.east);
\draw[redarrow] (halve.north) -- +(0,0.5) -| ([xshift=-4mm]phase.south);
\end{tikzpicture}}
\caption{\textbf{Enrollment for later verification.} A per-face key determines the signal embedded in the release. The fidelity gate checks the encoded image and reduces signal strength when needed, preserving facial appearance while preparing the release and private record for verification.}\vspace{-12mm}
\label{fig:embedding_pipeline}
\end{wrapfigure}

Fig.~\ref{fig:trigger_overview} links spectral agreement to face consistency and signed key projection to release-signal retention. Fig.~\ref{fig:embedding_pipeline} shows enrollment, from per-face key generation to encoding and visual-fidelity assessment.
\Needspace{5\baselineskip}
\subsection{Enrolling Faces for Release}
\label{sec:pct}
Enrollment prepares a released image $\mathbf I'$ and private record $R$ for later verification: the reference captures the enrolled face, while a weak keyed signal identifies its published version. The master seed and sample identifier derive a per-face key $s$ that determines the signal pattern.

\Needspace{4\baselineskip}
\noindent\textbf{Selecting where to place the signal.}
We use luminance to select one frequency ring and key pattern shared across RGB channels. Define $\mathbf Y_I=0.299\mathbf I_R+0.587\mathbf I_G+0.114\mathbf I_B$ and write its Fourier coefficients as $F_I(k)=A_I(k)e^{i\phi_I(k)}$. Here $k$ indexes a two-dimensional Fourier bin, $A_I$ and $\phi_I$ are its magnitude and phase, and $i$ in $e^{i\phi_I}$ is the imaginary unit. For horizontal and vertical frequencies $f_x,f_y$ in cycles per pixel, define $\kappa(k)=\sqrt{f_x^2+f_y^2}/\sqrt{0.5^2+0.5^2}$. A key $s$ seeds standard-normal draws $u_s(k)$ and the antisymmetric sign pattern $h_s(k)=\operatorname{sign}(u_s(k)-u_s(-k))$. The support is
\begin{equation}
\mathcal K_{\rm eff}=\{k:\kappa_{\min}\le\kappa(k)\le\kappa_{\max},\ A_I(k)>A_{\min},\ h_s(k)\ne0\},
\label{eq:annulus}
\end{equation}
The ring spans radii $\kappa_{\min}$ to $\kappa_{\max}$ and retains coefficients above the amplitude floor $A_{\min}=10^{-6}\max_k A_I(k)$ for phase comparison. The key pattern pairs each bin with its conjugate frequency $-k$, taken modulo the grid dimensions; $\operatorname{sign}$ returns $-1$, $0$, or $1$, assigning opposite signs to each active pair. Self-conjugate bins have $h_s=0$ and fall outside the active support. This luminance-derived support and pattern are shared across RGB channels.

\Needspace{4\baselineskip}
\noindent\textbf{Preserving visual fidelity.}
We control the signal with phase strength $\alpha$, in units of $\pi$ radians, and blend weight $\beta\in[0,1]$; $c\in\{R,G,B\}$ indexes color channels. For each channel, enrollment multiplies the active Fourier coefficients by $e^{i\alpha\pi h_s(k)}$, leaving their unblended magnitudes unchanged. The inverse transform is real because $h_s(-k)=-h_s(k)$. Writing $\widetilde F$ for the three edited channel spectra and $\mathcal F^{-1}$ for the inverse two-dimensional discrete Fourier transform, we blend the edited image with the original,
$\widetilde{\mathbf I}'=(1-\beta)\mathbf I+\beta\mathcal F^{-1}(\widetilde F)$,
where $\widetilde{\mathbf I}'$ denotes the unquantized blended image. We clip to $[0,255]$ and truncate to uint8. This signal-dependent quantization shapes the encoded release $\mathbf I'$; before it, the phase displacement is
\begin{equation}
\delta_s(k)=\arg\{(1-\beta)+\beta e^{i\alpha\pi h_s(k)}\}.
\label{eq:effective_delta}
\end{equation}
Here $\arg$ returns the principal complex phase in radians. The shared RGB multiplier also displaces luminance, with active magnitudes scaled by $\sqrt{1-2\beta(1-\beta)(1-\cos\alpha\pi)}$, a second-order change for small $\alpha$. To enforce visual fidelity, enrollment reduces strength until the encoded release meets PSNR $\ge35$\,dB and SSIM $\ge0.98$ against the original.

\Needspace{4\baselineskip}
\noindent\textbf{Retaining evidence for verification.}
For later comparison, the template $\mathcal T_I$ stores pre-quantization phase; subtracting the displacement recovers original phase in exact arithmetic. Using $\operatorname{wrap}(x)=(x+\pi)\bmod 2\pi-\pi$ to map angles to $[-\pi,\pi)$, we obtain, for $k\in\mathcal K_{\rm eff}$,
\begin{equation}
\mathcal T_I(k)=\operatorname{wrap}(\phi_I(k)+\delta_s(k)),\qquad
\widehat\phi_I(k)=\operatorname{wrap}(\mathcal T_I(k)-\delta_s(k)).
\label{eq:template}
\end{equation}
The record also stores original phase and amplitude, key, strengths, support, and crop geometry. Explicit original phase avoids reconstruction quantization error. Nine encoded-release projections---one full-ring and eight bandwise---normalize verification scores without reloading release pixels.

Before quantization, every active bin receives $\delta_s(k)=\pm\delta_{\rm eff}$. With $\beta=0.15$, $\alpha=0.05$ on FF++/Celeb-DF gives $\delta_{\rm eff}\simeq0.0235$\,rad; $\alpha=0.015$ on DFDC gives $0.0071$\,rad. Section~\ref{sec:verifier} uses this reference and signed displacement to assess face consistency and release-signal retention.

\vspace{-2mm}\subsection{Verifying Faces against Release Records}\vspace{-1mm}
\label{sec:verifier}
Direct verification combines template agreement $C$ for face consistency with normalized key projection $G_s^{(\mathrm{release})}$ for release-signal retention. For eligible queries, confirmation requires both scores to meet their calibrated thresholds. Original-reference agreement $H$ provides a comparison for analysis. Fig.~\ref{fig:tcrh} connects these scores to verification and review.

\begin{figure}[t]
\centering
\includegraphics[width=\textwidth]{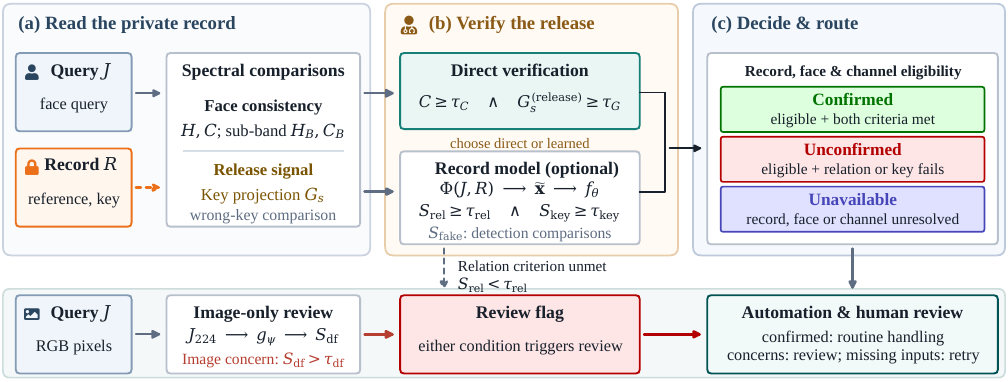}
\caption{\textbf{From release evidence to verification and review.} Direct thresholds and the record model offer two paths for assessing face consistency and release evidence. Eligibility and both verification criteria determine the status, guiding automated handling or human assistance. Optional image-based review adds a separate flag.}\vspace{-2mm}

\label{fig:tcrh}
\end{figure}
\noindent\textbf{Spectral agreement.}
Query luminance $\mathbf Y_J$, computed with the RGB weights of $\mathbf Y_I$, supplies Fourier phase $\phi_J$. We compare it with the stored original phase $\widehat\phi_I$ and release template $\mathcal T_I$. Let $r_k(\mathbf J)=\operatorname{wrap}(\phi_J(k)-\widehat\phi_I(k))$ be the query residual. Weights $w_k=A_I(k)/\max_j A_I(j)$ emphasize frequencies that were strong in the original image; $k$ and $j$ index frequency bins. We compute
\begin{equation}
H(\mathbf J;R)=\frac{\sum_{\mathcal K_{\rm eff}}w_k\cos r_k(\mathbf J)}{\sum_{\mathcal K_{\rm eff}}w_k},\qquad
C(\mathbf J;R)=\frac{\sum_{\mathcal K_{\rm eff}}w_k\cos(\phi_J(k)-\mathcal T_I(k))}{\sum_{\mathcal K_{\rm eff}}w_k}.
\label{eq:cosine_controls}
\end{equation}
Larger $H$ indicates agreement with the original face, while larger $C$ indicates agreement with the release template. Their complements $1-H$ and $1-C$ are the original-reference error and template mismatch. Since $H$ favors the original image and neither score tests a candidate key, high agreement alone does not identify the published version.

To exclude frequencies degraded by processing, a retained sub-band $B$ restricts the sums to $\mathcal K_{\rm eff}\cap B$, giving $H_B$ and $C_B$. Relation-only evaluation uses $H_B\ge\tau_H$ (original reference) or $C_B\ge\tau_C$ (release template), with its cutoff fitted to the first percentile of channel-processed benign training scores. Sub-band detection uses $\kappa\le0.15$; joint comparisons use full-ring $C$.

\noindent\textbf{Signed key projection.}
To resolve the version ambiguity, we project the residual onto its enrolled displacement $\delta_s$. This signed response captures small displacements at first order:
\begin{equation}
G_s(\mathbf J;R)=\frac{\sum_{k\in\mathcal K_{\rm eff}}w_k\delta_s(k)\sin r_k(\mathbf J)}{\sum_{k\in\mathcal K_{\rm eff}}w_k\delta_s(k)^2+\epsilon_0}.
\label{eq:key_projection}
\end{equation}
Here $\epsilon_0=10^{-12}$ stabilizes the denominator. An ideal unquantized release has $r_k=\delta_s(k)$ and $G_s\simeq1$, whereas the original has zero residual. Contributions from independent wrong keys cancel in expectation. The projection thus supplies release-version evidence that complements spectral agreement. App.~\ref{app:wrong_key_null} characterizes the empirical wrong-key distribution.

\noindent\textbf{Normalized key evidence.}
Direct verification measures retention relative to the enrolled release using $G_s^{(\mathrm{release})}=G_s(\mathbf J)/\max(|G_s(\mathbf I')|,10^{-6})$. It pairs this normalized projection with spectral agreement, calibrating a threshold for each criterion under the processing channel. Reference controls use the template projection as the denominator, denoted $G_s^{(\mathrm{template})}$.

\vspace{-2mm}
\subsection{Making Verification and Review Decisions}\vspace{-1mm}
\label{sec:policy}
The record model integrates spectral statistics into verification scores, while an image-only detector assesses manipulation from query pixels without access to the release record. The policy determines release consistency through direct thresholds or learned scores and uses a separate review flag to guide further assessment.

\noindent\textbf{Key-contrast features.}
The record model also uses the contrast between the enrolled key and $m=8$ wrong projection keys, evaluated on the same query and record. Writing $G_j^-$ for wrong-key response $j$, we compute
\begin{equation}
(\overline G^-,G^-_{\max},G_\Delta,G_z)=
\left(\tfrac1m\sum_jG_j^-,\max_jG_j^-,G_s-\max_jG_j^-,
\frac{G_s-\overline G^-}{\operatorname{std}_jG_j^-+10^{-8}}\right).
\label{eq:wrong_key_stats}
\end{equation}
These features use the stored release response for each full-ring or bandwise comparison to normalize both enrolled- and wrong-key projections. The bar and $\max$ denote mean and maximum, and $\operatorname{std}_j$ uses divisor $m$. The margin $G_\Delta$ compares the enrolled key with the strongest sampled wrong key; $G_z$ scales separation by their spread. Changing the projection key leaves $H$ and $C$ unchanged.

\noindent\textbf{Record model.}
The record model combines spectral agreement and key evidence across frequencies to learn their contribution to verification. Its input summarizes the query--record comparison through full-ring statistics $\mathbf f_0\in\mathbb R^{11}$ and eight radial-band feature vectors $\mathbf f_\ell\in\mathbb R^8$, $\ell=1,\ldots,8$. The fixed feature map
\begin{equation}
\mathbf x=\Phi(\mathbf J,R)
=[\mathbf f_0;\mathbf f_1;\cdots;\mathbf f_8]\in\mathbb R^{75}
\label{eq:record_features}
\end{equation}
concatenates agreement, phase-residual, normalized-projection, and wrong-key statistics; semicolons denote concatenation. App.~\ref{app:hyper} specifies the features. Training-set means $\boldsymbol\mu$ and standard deviations $\boldsymbol\sigma$ place them on comparable scales before prediction:
\begin{equation}
\widetilde{\mathbf x}=(\mathbf x-\boldsymbol\mu)\oslash\boldsymbol\sigma,
\qquad
(S_{\rm rel},S_{\rm key},S_{\rm fake})=f_\theta(\widetilde{\mathbf x}).
\label{eq:record_model_map}
\end{equation}
Here $\oslash$ denotes elementwise division, with near-zero scales replaced by one. The MLP $f_\theta$ has parameters $\theta$ and sigmoid outputs: $S_{\rm rel}$ scores face consistency, $S_{\rm key}$ scores correct-key release association, and $S_{\rm fake}$ scores re-rendering. The first two guide confirmation; the third supports detection comparisons.

\noindent\textbf{Training objectives.}
Training keeps enrollment and $\Phi$ fixed. Benign, swapped, and wrong-key queries have binary targets $(y_{\rm rel},y_{\rm key},y_{\rm fake})$ of $(1,1,0)$, $(0,0,1)$, and $(1,0,0)$, respectively. This supervision makes $S_{\rm key}$ a score for the correct-key release class, while $G_s$ measures the retained signal. We train the three heads jointly:
{\small
\begin{equation}
\mathcal L(\theta)=\mathrm{BCE}(S_{\rm fake},y_{\rm fake})+
\lambda_{\rm rel}\mathrm{BCE}(S_{\rm rel},y_{\rm rel})+
\lambda_{\rm key}\mathrm{BCE}(S_{\rm key},y_{\rm key})+
\lambda_{\rm agr}\mathbb E_{y_{\rm rel}=1}\|z_{\rm rel}-z_{\rm key}\|_2^2.
\label{eq:loss_core}
\end{equation}}
Here $\mathrm{BCE}$ denotes binary cross-entropy, and the $\lambda$ coefficients weight the objectives. The agreement term penalizes the distance between auxiliary projections $z_{\rm rel},z_{\rm key}\in\mathbb R^8$ on $y_{\rm rel}=1$ queries; all three classification heads use the shared hidden features. App.~\ref{app:hyper} details the architecture and evaluates how the agreement penalty affects verification by comparing training with and without this term.

\noindent\textbf{Independent image-based review.}
Image-based review adds manipulation cues from the query pixels to the record-based assessment. A detector $g_\psi$, trained separately by cross-entropy, maps the resized RGB query $\mathbf J_{224}\in[0,1]^{224\times224\times3}$ to a softmax score $S_{\rm df}=g_\psi(\mathbf J_{224})$. A vision-language model such as M2F2-Det~\citep{Guo2025M2F2Det} could provide diagnostic reports for reviewers to interpret suspected manipulation and guide further assessment, subject to latency and computational budgets.

\noindent\textbf{Decision policy.}
Eligibility determines whether verification can proceed. We set $\mathsf{eligible}(J)=1$ when the record is resolved, face association and geometry are valid, and the channel is supported; the selected path then compares scores with its thresholds. Otherwise, $\mathsf{eligible}(J)=0$ yields unavailable regardless of scores, indicating missing prerequisites rather than a failed verification criterion:
{\small
\begin{equation}
\mathsf{status}(J)=\left\{
\renewcommand{\arraystretch}{1.2}
\begin{array}{@{}ll@{\quad}l@{}}
\mathrm{confirmed}, & \mathsf{eligible}(J)=1,\ C\ge\tau_C\ \land\ G_s^{(\mathrm{release})}\ge\tau_G
& \multirow{2}{*}{\text{(direct, Sec.~\ref{sec:verifier})}}\\
\mathrm{unconfirmed}, & \mathsf{eligible}(J)=1,\ C<\tau_C\ \lor\ G_s^{(\mathrm{release})}<\tau_G & \\
\mathrm{confirmed}, & \mathsf{eligible}(J)=1,\ S_{\rm rel}\ge\tau_{\rm rel}\ \land\ S_{\rm key}\ge\tau_{\rm key}
& \multirow{2}{*}{\text{(learned, Sec.~\ref{sec:policy})}}\\
\mathrm{unconfirmed}, & \mathsf{eligible}(J)=1,\ S_{\rm rel}<\tau_{\rm rel}\ \lor\ S_{\rm key}<\tau_{\rm key} & \\
\mathrm{unavailable}, & \mathsf{eligible}(J)=0. &
\end{array}\right.
\label{eq:accept}
\end{equation}}
Calibration sets the tolerance for normal variation using benign training queries. Direct thresholds $\tau_C,\tau_G$ share a lower-tail quantile selected to keep joint calibration rejection at most $1\%$; using $C_B$ requires calibration on that sub-band. Learned thresholds $\tau_{\rm rel},\tau_{\rm key}$ use the first percentiles of target-dataset benign training scores with model weights fixed. These per-head cutoffs can yield joint rejection above $1\%$. App.~\ref{app:joint_policy} compares both paths at a matched joint calibration budget.

For human assessment of eligible queries, image-only review flags $S_{\rm df}>\tau_{\rm df}$; with the record model, $\widehat y_{\rm review}=\mathbf1[S_{\rm rel}<\tau_{\rm rel}]\lor\mathbf1[S_{\rm df}>\tau_{\rm df}]$ combines relation discrepancies and image-based concerns. Here $\tau_{\rm df}$ is the calibrated image-only cutoff, $\mathbf1$ an indicator, and $\land/\lor$ denote AND/OR. Together, statuses and flags guide routine handling, review of discrepancies, and requests for the records or processing information needed to complete verification and inform publisher decisions.

\vspace{-2mm}
\section{Experiments}
\label{sec:exp}
\label{sec:exp-setup}\vspace{-1mm}

\noindent\textbf{Datasets and enrollment.}
We enroll $7{,}366$ real faces from FF++, DFDC, and Celeb-DF~\citep{Rossler2019FFpp,Dolhansky2020DFDC,Li2020CelebDF}, using native-resolution crops located with SCRFD~\citep{Guo2022SCRFD} and cross-checked with YuNet or a Haar-cascade face detector. Each face receives a release and private record; verification reads that record without reloading the original.

\noindent\textbf{Queries and splits.}
To isolate manipulation effects, each released image is paired with a SimSwap~\citep{Chen2020SimSwap} face swap; both are verified against the same record, yielding $14{,}732$ queries. Replacement face identities come from manipulated images in the same dataset and split.

\noindent\textbf{Metrics and training settings.}
Face-swap detection AUC compares swaps with paired releases, excluding wrong-key queries; Table~\ref{tab:robust_detail_main} also covers StarGAN, UniFace, InSwapper, and SD inpainting. Key-discrimination AUC compares correct and wrong keys, while FPR measures benign rejection. Pass/rejection rates describe the specified criterion, not necessarily joint confirmation. Transfer tests freeze FF++- or DFDC-trained weights; App.~\ref{app:experimental_setup} gives preparation, calibration, and training details.

\vspace{-3mm}
\subsection{Enrollment: creating the release}\vspace{-1mm}
\label{sec:exp-enroll}
\label{sec:exp-fidelity}
We evaluate release distortion and sensitivity to face re-rendering on the same assets. The release must therefore preserve visual quality while providing evidence that tolerates benign processing and remains sensitive to changes in face content.

\noindent\textbf{Visual quality.}
Fig.~\ref{fig:imperceptibility} shows the fidelity of phase editing, blending, and uint8 encoding. On the same face, our release achieves $51.2$\,dB PSNR, compared with $29.9$ for StegaStamp and $41.4$ for KAD-Net. The private record makes the weak signal measurable against a known reference, supporting high-fidelity publication at the cost of retaining per-release evidence. Accordingly, the comparison reflects different payload and reference budgets. App.~\ref{app:fidelity_table} reports fidelity across datasets and resolutions.

\begin{figure}[!h]
\centering
\includegraphics[width=\textwidth]{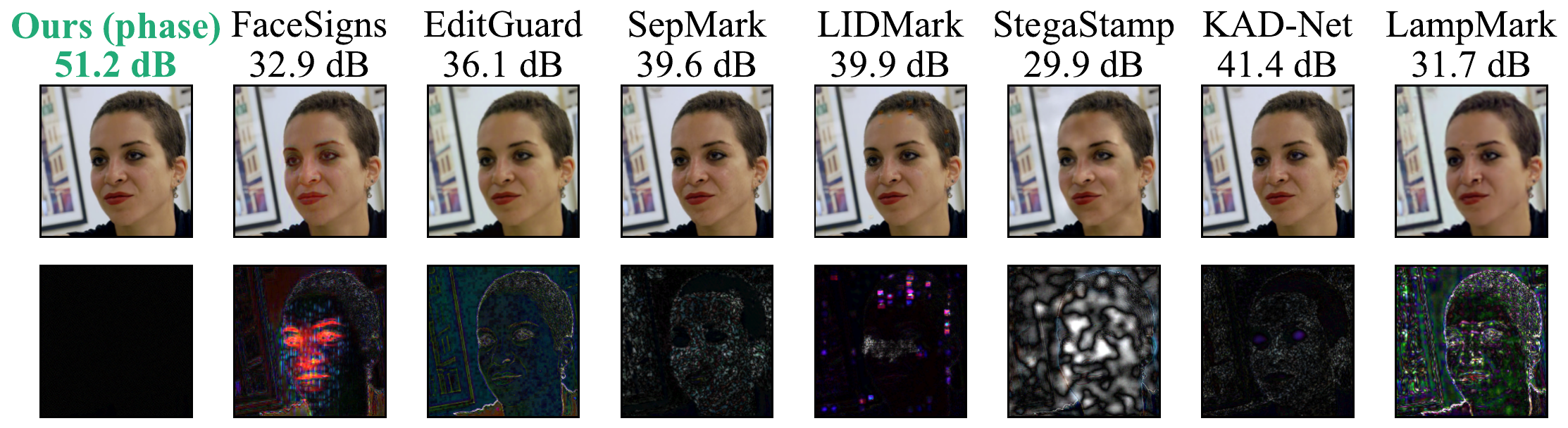}
\caption{\textbf{Visual fidelity of released images.} Phase editing, blending, and uint8 encoding yield small release distortion. Top: the same \textsc{FF++} face under each marking configuration, with PSNR relative to the original. Bottom: $10\times$-amplified residuals; payload and reference budgets differ across methods.}
\label{fig:imperceptibility}
\end{figure}

\vspace{-3mm}
\subsection{Face and Release Verification}
\label{sec:exp-verif}
\label{sec:exp-robust}\vspace{-1mm}

\begin{table}[tbp]
\centering\footnotesize\setlength{\tabcolsep}{2.5pt}\renewcommand{\arraystretch}{1.03}
\caption{\textbf{Frequency-restricted spectral agreement} (AUC $\uparrow$). Reading only $\kappa\le0.15$ ($C_B$) improves separation under benign processing, particularly blur, compared with the full ring ($C$). The FF++ comparison includes the band-diagnostic cohort and is exploratory (App.~\ref{app:readout}).}
\label{tab:robust_detail_main}
\begin{tabular*}{\linewidth}{@{\extracolsep{\fill}}l cccc ccc cccc@{}}
\toprule
& \multicolumn{4}{c}{Benign processing (FF++)} & \multicolumn{3}{c}{Clean SimSwap} & \multicolumn{4}{c}{Other face manipulations (FF++)} \\
\cmidrule(lr){2-5}\cmidrule(lr){6-8}\cmidrule(lr){9-12}
Direct score & JPEG & Blur & WebP & Pixel. & \textsc{FF++} & \textsc{DFDC} & \textsc{CDF} & StarGAN & UniFace & InSwap & SD-inp \\
\midrule
Template $1-C_B$ & \textbf{0.999} & \textbf{0.936} & \textbf{0.995} & \textbf{1.000} & \textbf{1.000} & \textbf{0.992} & \textbf{0.991} & \textbf{1.000} & \textbf{1.000} & \textbf{1.000} & \textbf{1.000} \\
\quad $1-C$ (full ring) & 0.994 & 0.597 & 0.979 & 0.996 & \textbf{1.000} & 0.978 & 0.984 & \textbf{1.000} & \textbf{1.000} & \textbf{1.000} & \textbf{1.000} \\
\bottomrule
\end{tabular*}
\end{table}

\noindent\textbf{Complementary relation and key evidence.}
Fig.~\ref{fig:signal} shows distinct responses to face re-rendering and projection-key changes. Evaluated re-renderings increase original-reference error $1-H$, while key projection can remain high, particularly on Celeb-DF and DFDC. Changing only the projection key instead suppresses the key response and leaves $H$ unchanged. The two scores thus provide separate evidence of face consistency and release-signal retention for the joint decision. Dataset-specific calibration accounts for variation in their distributions, as examined in App.~\ref{app:evt_transfer}.

\begin{figure}[!h]
\centering
\includegraphics[width=\textwidth]{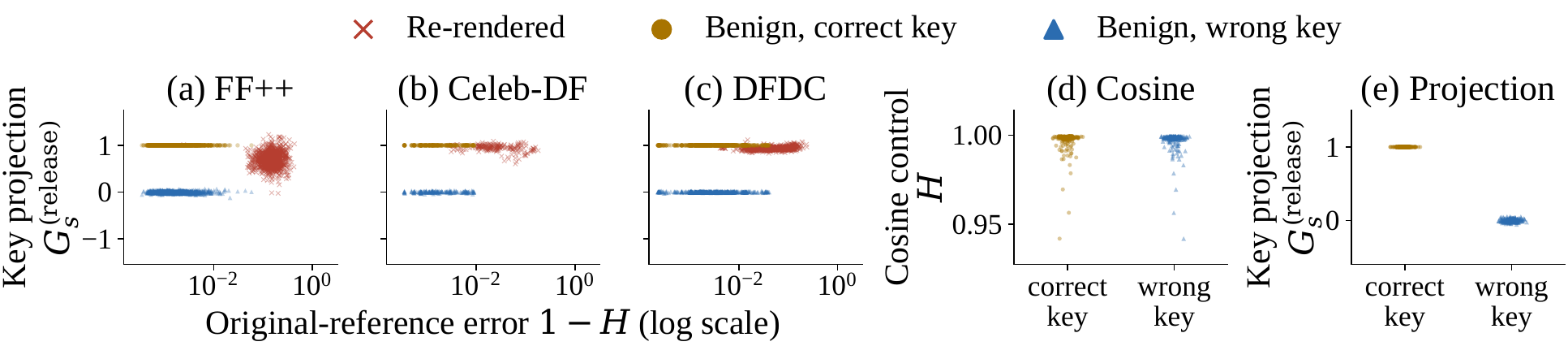}\vspace{-2mm}
\caption{\textbf{Complementary evidence for face consistency and release-signal retention.} Left: re-rendering changes the face spectrum and raises original-reference error $1-H$, while the keyed signal can persist. Right: a wrong projection key disrupts alignment with the release signal (e) without changing the face comparison (d). Joint verification therefore combines spectral agreement with release-normalized key projection to assess both aspects of the enrolled release.}
\label{fig:signal}\vspace{-3mm}
\end{figure}

\noindent\textbf{Benign channels.}
Original-image amplitude weights emphasize strong enrolled coefficients, while $C_B$ concentrates verification on the lower frequencies preserved by processing. Table~\ref{tab:robust_detail_main} shows how this design improves face-swap separation. A fixed-band evaluation on additional source videos further supports its effectiveness under blur and reposting. Channel selection nevertheless requires reliable benign acceptance and manipulation rejection at the deployed threshold; App.~\ref{app:readout} examines these operating points and the retention of release-key evidence.

\noindent\textbf{Comparison with peer methods.}
\label{sec:exp-peers}
Each watermarking peer embeds its own mark into the same aligned source crops and undergoes SimSwap with the same replacement identities. Table~\ref{tab:peer_detect} compares their verification signals. The methods differ in what each verifier reads and what its verdict establishes:
\begin{itemize}[leftmargin=0pt]\vspace{-1mm}
\item \textbf{Manifest.} C2PA verifies signed assertions about an asset's provenance~\citep{c2pa}, but does not itself compare a circulating face with a publisher's private enrollment reference.\vspace{-1mm}
\item \textbf{Embedded marks.} Robust message recovery~\citep{Tancik2020StegaStamp,Sander2025WAM,Lu2025VINE} can succeed after a face swap, leaving little evidence of replacement: WAM and VINE-R yield AUC $0.50$. Fragile marks and decoder mismatch~\citep{Neekhara2024FaceSigns,Wu2023SepMark,He2026KADNet}, or landmark displacement~\citep{Wu2026LIDMark}, are more directly sensitive to manipulation, but their effectiveness varies across datasets in Table~\ref{tab:peer_detect}. A recovered mark alone therefore cannot establish consistency with the enrolled face.\vspace{-1mm}
\item \textbf{Image detectors.} Recent detectors learn manipulation cues through discrepancy learning and vision-language representations~\citep{Yang2025D3,Guo2025M2F2Det,Jung2026SD2}. Our ResNet-50 control reaches AUC $0.68$ on Celeb-DF and $0.76$ on DFDC with fixed FF++ weights, illustrating the importance of target-dataset evaluation. In our framework, image-only scores guide review alongside verification against the private release record.\vspace{-1mm}
\item \textbf{Private face record.} Our verifier compares the query with the enrolled release and reads the release-specific keyed displacement under calibrated benign channels. The record-model row reports $S_{\rm fake}$; release confirmation additionally requires both relation and key criteria to be satisfied.
\end{itemize}

\begin{table}[!htb]

\centering\fontsize{8}{8.8}\selectfont\setlength{\tabcolsep}{2.4pt}\vspace{-3mm}
\caption{\textbf{Face-swap detection within and across datasets} (AUC $\uparrow$). Both record-based paths separate swaps across datasets. Learned scores reach near-perfect within-dataset AUC, while direct comparison transfers better to Celeb-DF. $\dagger$: fixed public weights; ---: no dataset-specific training.}
\label{tab:peer_detect}
\label{tab:main}
\begin{tabular*}{\linewidth}{@{\extracolsep{\fill}}ll c ccc cc cc@{}}
\toprule
 & & & \multicolumn{3}{c}{Per-dataset evaluation} & \multicolumn{2}{c}{From \textsc{FF++}} & \multicolumn{2}{c}{From \textsc{DFDC}} \\
\cmidrule(lr){4-6}\cmidrule(lr){7-8}\cmidrule(lr){9-10}
Method & Detection score & Bits & \textsc{FF++} & \textsc{CDF} & \textsc{DFDC} & $\to$\textsc{CDF} & $\to$\textsc{DFDC} & $\to$\textsc{CDF} & $\to$\textsc{FF++} \\
\midrule
StegaStamp & Robust bit error & 100 & 0.91 & 0.57 & 0.52 & 0.59 & 0.54 & 0.56 & 0.80 \\
FaceSigns$^\dagger$ & Fragile bit error & 128 & \textbf{1.00} & 0.70 & 0.70 & --- & --- & --- & --- \\
SepMark$^\dagger$ & Decoder mismatch & 128 & 0.69 & 0.50 & 0.50 & --- & --- & --- & --- \\
WAM$^\dagger$ & Bit error & 32 & 0.50 & 0.50 & 0.50 & --- & --- & --- & --- \\
KAD-Net$^\dagger$ & Fragile bit error & 128 & 0.54 & 0.50 & 0.50 & --- & --- & --- & --- \\
VINE-R$^\dagger$ & Bit error & 100 & 0.50 & 0.50 & 0.50 & --- & --- & --- & --- \\
LIDMark & Landmark distance (px) & 152 & \textbf{1.00} & 0.74 & 0.84 & 0.90 & 0.74 & 0.92 & 0.98 \\
\midrule
\multirow{3}{*}{Ours} & Image-only $S_{\rm df}$ & none & 0.99 & 0.98 & \textbf{1.00} & 0.68 & 0.76 & 0.93 & 0.95 \\
                         & Direct spectral $1-C_B$ & none & \textbf{1.00} & 0.99 & 0.99 & \textbf{0.99} & 0.99 & \textbf{0.99} & \textbf{1.00} \\
                         & Record model $S_{\rm fake}$ & none & \textbf{1.00} & \textbf{1.00} & \textbf{1.00} & 0.98 & \textbf{1.00} & 0.97 & \textbf{1.00} \\
\bottomrule
\end{tabular*}

\end{table}

\begin{wraptable}{r}{0.48\textwidth}
\centering\footnotesize\setlength{\tabcolsep}{2pt}
\vspace{-5mm}
\caption{\textbf{Release evidence from the key.} Direct reference tests detect swaps but admit never-released originals under JPEG calibration. Key projection $G_s$ distinguishes releases from originals even after JPEG processing.}
\label{tab:reference_controls_main}
\begin{tabular}{@{}l ccc cc@{}}
\toprule
& \multicolumn{3}{c}{Clean swap AUC} & \multicolumn{2}{c}{Pass rate (\%)} \\
\cmidrule(lr){2-4}\cmidrule(lr){5-6}
Direct test & FF++ & DFDC & CDF & \shortstack{Original\\pass $\downarrow$} & \shortstack{Release\\pass $\uparrow$} \\
\midrule
Original $H$ & 1.000 & 0.977 & 0.978 & 100 & 99.25 \\
Template $C$ & 1.000 & 0.978 & 0.984 & 100 & 99.25 \\
SSIM & 1.000 & 1.000 & 1.000 & 100 & 99.5 \\
ArcFace & 1.000 & 1.000 & 1.000 & 99.25 & 99.0 \\
Key $G_s$ & -- & -- & -- & 0 & 99.5 \\
\bottomrule
\end{tabular}\vspace{-5mm}
\end{wraptable}
\noindent\textbf{Effect of key projection.}
With thresholds calibrated on each dataset's benign training releases, Table~\ref{tab:reference_controls_main} shows that references admit most unreleased originals, whereas $G_s$ rejects them while accepting JPEG releases. The key therefore helps distinguish the published version from its original.

\noindent\textbf{Joint decision policy across datasets.}
At a matched $1\%$ joint calibration budget, direct verification rejects more transferred swaps, and both paths reject all tested wrong keys. Under JPEG and reposting, the record model accepts more benign releases but also more swaps. Publishers can therefore start with direct verification and use the record model when reducing benign rejection is valuable and target-channel swap rejection remains adequate. App.~\ref{app:joint_policy} quantifies this tradeoff.

\noindent\textbf{Verification cost and storage.}
Relation-only $C_B$ takes $5.9$\,ms in memory; NPZ loading, full-ring features, and MLP take $149.9$\,ms. Separately, a $2.3$\,KiB sub-band payload preserves detection/key AUCs to three decimals on $150$ FF++ assets, excluding keys and metadata. App.~\ref{app:deploy_cost} specifies the timed paths and codec configurations.

\subsection{Verification after Image Processing}

\label{sec:exp-roundtrip}
To examine how delivery affects verification, we process $200$ FF++ releases and their SimSwap, UniFace reenactment, and SD inpainting counterparts through each channel. Thresholds follow Sec.~\ref{sec:verifier}, using all $1{,}144$ benign training releases. Providing the enrolled records isolates spectral robustness, while credential retention measures whether the record link survives.

\begin{figure}[!htbp]
\centering\vspace{-2mm}
\includegraphics[width=0.92\textwidth]{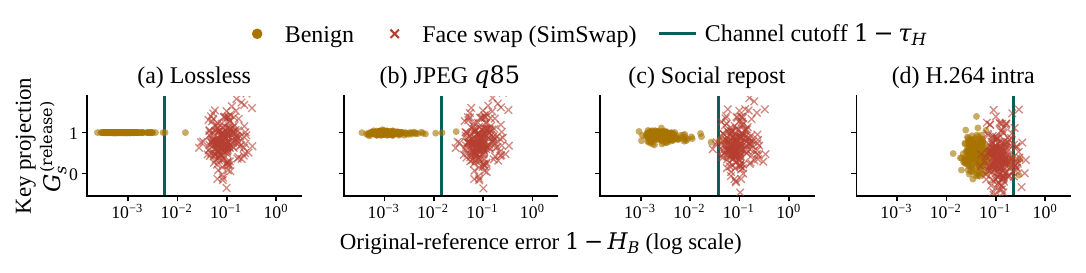}
\vspace{-5mm}
\caption{\textbf{Direct spectral evidence across channels.} Vertical lines mark $1-\tau_H$ for channel-calibrated agreement $H_B$; release-normalized key projection is shown separately. Table~\ref{tab:roundtrip} reports channel-calibrated pass/reject rates.}
\label{fig:roundtrip}\vspace{-2mm}
\end{figure}

Fig.~\ref{fig:roundtrip} and Table~\ref{tab:roundtrip} show strong manipulation separation after lossless read-back, metadata stripping, screenshots, JPEG $q85$, and one repost. However, repeated reposting, intra-frame H.264, and cropping weaken swap rejection despite high benign pass rates; reenactment and SD inpainting remain more separable. A benign channel $\mathcal P$ can therefore erase evidence distinguishing the histories in Eq.~\eqref{eq:hypotheses}. The supported set $\mathcal C$ must reflect verification performance, not merely benign processing: recalibration cannot restore lost separation. Clip averaging partly recovers swap rejection, as detailed in App.~\ref{app:readout}; more stable features and geometry-aware comparison may further extend channel coverage.

\begin{table}[!htb]
\centering\scriptsize\setlength{\tabcolsep}{2.2pt}\vspace{-2mm}
\caption{\textbf{Credential retention and direct spectral agreement across delivery channels.} Relation tests use $H_B$ for stills, masked templates for crops, and mean $C_B$ for $120$ SimSwap clips. Channel cutoffs target $1\%$ benign rejection. Credential retention requires an HMAC-authenticated asset ID.}
\label{tab:roundtrip}
\begin{tabular*}{\linewidth}{@{\extracolsep{\fill}}l c r r r r@{}}
\toprule
 & & \multicolumn{4}{c}{\shortstack{Channel-calibrated relation tests\\target benign rejection: $1\%$}} \\
\cmidrule(lr){3-6}
Channel & \shortstack{Credential\\retained} & \shortstack{Face swap\\reject $\uparrow$} & \shortstack{Reenactment\\reject $\uparrow$} & \shortstack{SD inpaint.\\reject $\uparrow$} & \shortstack{Benign\\pass $\uparrow$} \\
\midrule
\csname @@input\endcsname tables/table4_layout_rows.tex
\bottomrule
\end{tabular*}\vspace{-2mm}

\end{table}

\noindent\textbf{Takeaway.} Because processing affects relation and key evidence differently, channel selection requires evaluating both signals together. Joint tests in App.~\ref{app:joint_policy} show that reposting increases swap acceptance for both paths. Publishers should therefore validate joint decisions on intended delivery channels and provide a trusted record lookup when credentials are lost.

\vspace{-2mm}
\section{Conclusion}
\label{sec:conclusion}\vspace{-1mm}
\method enables publishers to assess circulating faces against specific enrolled releases. A private spectral reference and weak keyed signal provide complementary evidence of face consistency and release-signal retention while preserving visual fidelity. Experiments across three datasets support direct verification without a trained model under calibrated processing conditions, making enrollment a basis for subsequent verification.

More broadly, \method connects verification to publisher decision-making: confirmed queries support routine handling, unconfirmed queries identify discrepancies for review, and unavailable queries prompt additional information or human assistance. These outcomes could help publishers manage verification at scale and prioritize further assessment. Linking evidence to release records, the framework complements provenance credentials and image-only detection.

\clearpage
\subsection*{Ethics Statement}
Publishers such as news organizations and public institutions could retain enrollment records, while publisher-operated verification services assess copies submitted by users or reviewers. Deployment requires enrollment consent, controlled access to private records, and retention and deletion policies. Verification results could support review of circulating images, but public reports should clearly communicate coverage and uncertainty: consistency with a released version does not establish the truth of the depicted event. Private spectral records contain sensitive face-derived information and require protection, as examined in App.~\ref{app:metadata_privacy}.

\subsection*{AI use statement}
Generative AI tools assisted with editing and polishing the manuscript, drafting and revising portions of the text, identifying related work, checking bibliographic information, and developing mathematical proofs and their written presentation. AI assistance also supported code revision, execution of supplementary experiments, analysis of results, and consistency checks across the method description, reported measurements, and released code. These activities informed revisions to the interpretation and scope of empirical claims. The authors retain responsibility for the correctness, originality, and final presentation of all text, proofs, results, and code.

\bibliographystyle{iclr2027_conference}
\bibliography{main}

\newpage
\section*{Appendix Overview}
\renewcommand{\arraystretch}{0.95}
\noindent\begin{tabular}{@{}l@{\hspace{.6em}}p{0.76\linewidth}r@{}}
\multicolumn{3}{@{}l}{\textbf{Background and Analysis}} \\
\S\,\ref*{app:notation} & \hyperref[app:notation]{Notation}\dotfill & \pageref{app:notation} \\
\S\,\ref*{sec:related} & \hyperref[sec:related]{Related Work}\dotfill & \pageref{sec:related} \\
\S\,\ref*{sec:score_model} & \hyperref[sec:score_model]{Analysis of Relation and Key Evidence}\dotfill & \pageref{sec:score_model} \\
\multicolumn{3}{@{}l}{\textbf{Experimental Evidence}} \\
\S\,\ref*{app:experimental_setup} & \hyperref[app:experimental_setup]{Experimental Setup and Reproducibility}\dotfill & \pageref{app:experimental_setup} \\
\S\,\ref*{app:fidelity_table} & \hyperref[app:fidelity_table]{Release Fidelity and Visual Examples}\dotfill & \pageref{app:fidelity_table} \\
\S\,\ref*{app:reference_controls} & \hyperref[app:reference_controls]{Private References and Release-Key Evidence}\dotfill & \pageref{app:reference_controls} \\
\S\,\ref*{app:robust_table} & \hyperref[app:robust_table]{Channel Robustness and Spectral Verification}\dotfill & \pageref{app:robust_table} \\
\S\,\ref*{app:exp_detail} & \hyperref[app:exp_detail]{Operating Points and Calibration}\dotfill & \pageref{app:exp_detail} \\
\S\,\ref*{app:sensitivity} & \hyperref[app:sensitivity]{Sensitivity of the Release Mechanism}\dotfill & \pageref{app:sensitivity} \\
\multicolumn{3}{@{}l}{\textbf{Deployment and Failure Cases}} \\
\S\,\ref*{app:deploy} & \hyperref[app:deploy]{Deployment Interface and Cost}\dotfill & \pageref{app:deploy} \\
\S\,\ref*{app:limits} & \hyperref[app:limits]{Verification Limits and Record Privacy}\dotfill & \pageref{app:limits} \\
\end{tabular}
\renewcommand{\arraystretch}{1.0}
\medskip
\renewcommand{\thesection}{\Alph{section}}
\renewcommand{\theHsection}{appendix.\Alph{section}}
\renewcommand{\theHsubsection}{appendix.\Alph{section}.\arabic{subsection}}
\renewcommand{\theHsubsubsection}{appendix.\Alph{section}.\arabic{subsection}.\arabic{subsubsection}}
\setcounter{section}{0}

\section{Notation}
\label{app:notation}
Symbols are defined where they enter the method or analysis and collected here for reference. Bold symbols denote images or feature vectors; calligraphic symbols denote sets, operators, hypotheses, or the phase template. Subscripts $I$ and $J$ identify original and query quantities. All angles are in radians.

\renewcommand{\arraystretch}{1.12}
\begin{longtable}{@{}p{3.05cm}@{\hspace{1.2em}}p{\dimexpr\linewidth-3.05cm-1.2em\relax}@{}}
\toprule
\endhead

\multicolumn{2}{@{}l}{\textbf{Images and the record}} \\
$N_h,N_w$ & crop height and width in pixels; RGB supplies the third dimension \\
$\mathbf{I}$ & aligned original crop, before publication (Sec.~\ref{sec:setup}) \\
$\mathbf{I}'$ & released crop, the original after the phase displacement \\
$\mathbf{J}$ & query: the aligned crop received at verification \\
$s$ & secret key of one enrolled face \\
$R,R_i,R_{a,i}$ & $R$ is the record for a single crop; $R_i$ indexes a face and $R_{a,i}$ additionally indexes its asset $a$. The private record contains phase and amplitude arrays, support, key, strengths, and release projections (Sec.~\ref{sec:pct}) \\
$\mathcal{T}_{I},\ \mathcal{T}_{I}(k)$ & stored phase template: luminance phase after the ideal edit, before release quantization \\
$A_I(k),\ \phi_I(k)$ & original luminance amplitude and phase; $F_I(k)=A_I(k)e^{i\phi_I(k)}$ \\
$\widehat{\phi}_{I},\ \widehat{\phi}_I(k)$ & original phase obtained from the record, stored explicitly in the evaluated implementation and recoverable in exact arithmetic by Eq.~\eqref{eq:template} \\
$\phi_{J},\ \phi_J(k)$ & query luminance phase, as an array and at one bin \\
$\mathbf Y_I,\mathbf Y_J$ & original/query luminance, $0.299R+0.587G+0.114B$ for the respective RGB pixels \\
$F_I,\widetilde F$ & original luminance spectrum and the three phase-edited RGB spectra, respectively \\
$\widetilde{\mathbf I}'$ & blended release before clipping and uint8 truncation \\
$c,I_c$ & color-channel index $c\in\{R,G,B\}$ and the corresponding original channel \\
$\mathcal{F},\mathcal{F}^{-1}$ & two-dimensional discrete Fourier transform and its inverse \\
$P,U,V$ & publisher, requesting user, and publisher-side verifier in Fig.~\ref{fig:positioning} \\

\addlinespace[2pt]
\multicolumn{2}{@{}l}{\textbf{The ring and the displacement}} \\
$k,-k$ & a Fourier-bin index and its conjugate index, modulo the grid dimensions \\
$f_x,f_y$ & horizontal and vertical frequency in cycles per pixel \\
$B$ & selected frequency sub-band; direct comparison uses $B=\{k:\kappa(k)\le0.15\}$ \\
$\mathcal{K}_{\rm eff}$ & active bins, Eq.~\eqref{eq:annulus} \\
$\kappa$ & normalized radius of a Fourier bin; $\kappa_{\min},\kappa_{\max}$ bound the ring and $A_{\min}$, a $10^{-6}$ fraction of the largest amplitude, is the amplitude floor, together defining $\mathcal{K}_{\rm eff}$ \\
$u_s(k)$ & standard-normal field over the Fourier grid, drawn from a generator seeded with the key $s$ \\
$h_s(k),h$ & the key's sign pattern $\operatorname{sign}(u_s(k)-u_s(-k))$: $\pm1$ on every active bin, the same in every channel, with $h_s(-k)=-h_s(k)$ \\
$\Delta\phi_{s,k}$ & pre-blend phase challenge $\alpha\pi h_s(k)$, before the release blend \\
$\delta_s(k),\delta_i,\delta$ & signed keyed displacement; $i=k$ and bare $\delta$ abbreviate the binwise quantity in the analysis/figures. Its effective value is given by Eq.~\eqref{eq:effective_delta} \\
$\alpha,\beta$ & dimensionless phase strength (pre-blend angle $\alpha\pi$) and blend weight; $\alpha_i,\beta_i$ are the values for face $i$ \\
$\delta_{\rm eff}$ & displacement magnitude on every active bin, determined by $\alpha$ and $\beta$ (Sec.~\ref{sec:pct}) \\
$w_k$ & original-image amplitude weight on a bin, Eq.~\eqref{eq:key_projection} \\
$r_k(\mathbf{J}),r_i,r$ & wrapped phase drift of the query against $\widehat{\phi}_{I}$ (Sec.~\ref{sec:verifier}) \\

\addlinespace[2pt]
\multicolumn{2}{@{}l}{\textbf{Readings from the record}} \\
$H(\mathbf{J};R)$ & amplitude-weighted agreement with the original phase, Eq.~\eqref{eq:cosine_controls} \\
$H_B,C_B$ & the original/template readings restricted to sub-band $B$ \\
$1-H,1-H_B$ & original-reference error over the full ring or sub-band; larger values indicate poorer agreement \\
$C(\mathbf{J};R)$ & template agreement: the same reading against $\mathcal{T}_I$ instead \\
$1-C,1-C_B$ & template mismatch over the full ring or sub-band; distinct from original-reference error \\
$G_s(\mathbf{J};R)$ & key projection: signed first-order agreement with the key, Eq.~\eqref{eq:key_projection} \\
$s^{-},m,j$ & independently sampled wrong key, number of wrong keys, and their index $j=1,\ldots,m$; $j$ is a dummy bin index in the amplitude maximum; $i,j$ in Fig.~\ref{fig:template} instead index queries and records \\
$G^{-}_{j},\overline{G^{-}},G^{-}_{\max}$ & the same projection under $m$ independent wrong keys, their mean and maximum, Eq.~\eqref{eq:wrong_key_stats} \\
$G_{\Delta},G_z$ & key margin above the strongest sampled wrong key, and standardized separation from the wrong-key mean, Eq.~\eqref{eq:wrong_key_stats} \\
$\epsilon_0$ & projection denominator stabilizer, $10^{-12}$; distinct from the attack budget $\epsilon$ \\
$G_s^{(\mathrm{release})}$ & $G_s$ divided by the release's own stored projection (Sec.~\ref{sec:verifier}) \\
$G_s^{(\mathrm{template})}$ & $G_s$ divided by the template's projection instead, used by the reference controls \\

\addlinespace[2pt]
\multicolumn{2}{@{}l}{\textbf{Scores, labels, and the policy}} \\
$\Phi,\mathbf x$ & fixed query-and-record feature map and its 75-dimensional output, Eq.~\eqref{eq:record_features} \\
$\mathbf f_0,\mathbf f_\ell,\ell$ & 11 full-ring features, eight features from radial band $\ell$, and band index $\ell=1,\ldots,8$ \\
$[\cdot;\cdot],\oslash$ & vector concatenation and elementwise division \\
$\widetilde{\mathbf x},\boldsymbol\mu,\boldsymbol\sigma$ & standardized model input, training feature means, and training standard deviations with near-zero scales replaced by one \\
$f_\theta,\theta$ & learned record MLP and its parameters; maps $\widetilde{\mathbf x}$ to three sigmoid scores \\
$\mathbf J_{224},g_\psi,\psi$ & resized RGB query in $[0,1]^{224\times224\times3}$, image-only classifier, and its parameters; $g_\psi$ outputs $S_{\rm df}$ \\

$S_{\rm rel},S_{\rm key},S_{\rm fake}$ & sigmoid record-model scores for face consistency, a benign release paired with its correct key, and re-rendering; larger values favor these training targets (Sec.~\ref{sec:policy}) \\
$S_{\rm df}$ & image-only manipulation score, from $\mathbf{J}$ alone; larger is more suspicious \\
$y_{\rm rel},y_{\rm key},y_{\rm fake}$ & binary training labels assigned by query construction, with $1$ indicating the named target (App.~\ref{app:random_variables}) \\
\addlinespace[1pt]
$\lambda_{\rm rel},\lambda_{\rm key},\lambda_{\rm agr},\mathcal{L}$ & weights of relation BCE, key BCE, and representation agreement, and the record-model training objective $\mathcal L(\theta)$, Eq.~\eqref{eq:loss_core} \\
$z_{\rm rel},z_{\rm key}$ & eight-dimensional representations from linear maps of the hidden layer, pulled together for training rows with $y_{\rm rel}=1$, Eq.~\eqref{eq:loss_core} \\
$\tau_H,\tau_C$ & direct agreement cutoffs fitted separately for the selected reference, support, channel, and aggregation; relation-only tests use first-percentile benign-training calibration, while joint decisions use the budget in Sec.~\ref{sec:policy} \\
$\tau_G$ & cutoff on release-normalized key projection in the direct joint policy \\
$\tau_{\rm rel},\tau_{\rm key},\tau_{\rm df}$ & record and review thresholds in the decision policy (Sec.~\ref{sec:policy}) \\
$\tau$ & generic score cutoff; the theory applies it to $H$, while empirical tables specify the reading and calibration protocol \\
$\mathsf{eligible}(J),\mathsf{status}(J)$ & query eligibility indicator and three-state verdict for its associated record and channel, Eq.~\eqref{eq:accept} \\
$\widehat y_{\rm review}$ & binary review flag, separate from verification status; the learned-path form combines relation failure and image-only concern \\

\addlinespace[2pt]
\multicolumn{2}{@{}l}{\textbf{Channels and hypotheses}} \\
$\mathcal{C}$ & the benign channels a deployment supports \\
$\mathcal{P}$ & the processing channel taken by a query, $\mathcal{P}\in\mathcal{C}$ under $\mathcal{H}_c$ \\
$\mathcal{R},M$ & the set of face re-renderings, and one re-rendering in that set \\
$\mathcal{H}_c,\mathcal{H}_d$ & benign-channel and re-rendering-then-channel hypotheses, Eq.~\eqref{eq:hypotheses} \\

\addlinespace[2pt]
\multicolumn{2}{@{}l}{\textbf{Local noise and concentration analysis}} \\
$i,q_i,a_i$ & active-bin index $i=k$, $q_i=\cos r_i$, and normalized weight $a_i=w_i/\sum_jw_j$; this $i$ is not a face index \\
$b,B_b,a_b,Q_b$ & block index, its bin set, total weight $\sum_{i\in B_b}a_i$, and weighted mean cosine within the block \\
$n_w,n_{\rm eff}$ & effective sample sizes for independent weighted bins and independent weighted blocks, respectively \\
$\mu_c,b_a,\mu_d$ & nominal benign mean, nonnegative allowance for benign mean reduction, and upper bound on the re-rendered mean; $b_a$ is a scalar allowance, not a block weight \\
$t$ & positive score deviation in the concentration proof; in Fig.~\ref{fig:hyper}(a), the target benign rejection rate \\
$\xi_i$ & residual channel noise at bin $i$ in the local model $r_i=\delta_i+\xi_i$ \\
$\rho$ & attenuation factor $\sum_iw_i\delta_i^2/(\sum_iw_i\delta_i^2+\epsilon_0)$ in the projection expansion \\

\addlinespace[2pt]
\multicolumn{2}{@{}l}{\textbf{Operators and mathematical conventions}} \\
$\mathbb R,i$ & real numbers and the imaginary unit in complex exponentials; elsewhere $i$ is an explicitly scoped index \\
$\arg,\operatorname{wrap},\operatorname{sign}$ & principal complex argument, wrapping to $[-\pi,\pi)$, and sign in $\{-1,0,1\}$ \\
$\mathbb E,\Pr,\operatorname{std}_j$ & expectation, probability, and standard deviation across wrong keys with divisor $m$ \\
$\mathrm{BCE},\|\cdot\|_2,O(\cdot)$ & binary cross-entropy, Euclidean norm, and asymptotic order of a remainder \\
$\mathbf1[\cdot],\land,\lor$ & indicator (one when its condition holds), logical AND, and logical OR \\
$\cap,\cup,\in$ & set intersection, union, and membership \\

\addlinespace[2pt]
\multicolumn{2}{@{}l}{\textbf{Evaluation and figure conventions}} \\
$n$ & number of observations for the stated experiment; captions specify the counting unit \\
$q,p$ & JPEG quality setting and augmentation application probability; $q_i$ in the theory instead denotes a cosine reading \\
$\epsilon,L_\infty$ & attack budget ($8/255$) and maximum absolute pixel-change norm on images scaled to $[0,1]$ \\
$\uparrow,\downarrow,\pm$ & larger/smaller is preferable; mean plus/minus standard deviation where specified by the caption \\
PSNR, SSIM & peak signal-to-noise ratio in decibels and structural similarity, relative to the original/reference image \\
ROC, AUC & receiver operating characteristic and area under that curve \\
TPR, FPR, FRR & true-positive rate, false-positive rate, and benign false-rejection rate \\
JPEG $q90$, blur $k7$ & JPEG quality 90 and Gaussian blur with a $7\times7$ kernel; $k7$ is a channel label, not Fourier-bin index $k$ \\
FF++, DFDC, CDF & FaceForensics++, Deepfake Detection Challenge, and Celeb-DF v2; CelebDF is also used as a dataset label \\
CRF, QP & constant-rate-factor setting and quantization parameter in video-codec conditions \\

\bottomrule
\end{longtable}
\renewcommand{\arraystretch}{1.0}

\section{Related Work}
\label{sec:related}

\noindent\textbf{Provenance credentials.}
C2PA Content Credentials bind signed assertions about an asset's issuer and history to its content~\citep{c2pa}. Their validation depends on certificate handling, timestamps, and validator behavior, as studied for C2PA v2.2~\citep{Golaszewski2026C2PA}. Content-carried signals can reconnect circulating copies with provenance information after metadata loss~\citep{Simmons2024Interop}. \method complements this infrastructure with a private reference for face-level comparison against an enrolled release.

\noindent\textbf{Reference-based face verification.}
Identity embeddings such as ArcFace assess whether two faces depict the same person~\citep{Deng2019ArcFace}. Reference-based forensics also models person-specific behavior: ID-Reveal learns temporal facial signatures~\citep{Cozzolino2021IDReveal}, while Referee uses a reference video as an audiovisual identity anchor~\citep{Boo2026Referee}. More recently, regional identity and texture comparisons address identity-preserving generative manipulations~\citep{Ming2026Reference}. These approaches establish an important reference-based alternative to query-only detection. \method instead associates the reference with a particular released asset and combines face consistency with evidence of its keyed release signal.

\noindent\textbf{Content authentication.}
Fragile watermarking detects content changes, while semi-fragile designs tolerate specified processing such as JPEG compression~\citep{Fridrich1999Tamper,Lin2000SemiFragile}. Recent methods couple authentication evidence to image content: MetaSeal uses content-dependent cryptographic watermarks for attribution~\citep{Zhou2026MetaSeal}, and DeepSignature combines signed content-encoding watermarks with latent-space tamper verification~\citep{Graf2026DeepSignature}. Thus, content-dependent verification is an established direction. Our formulation centers on publisher-held spectral records and separate face-relation and release-key evidence under calibrated processing; the reference controls assess what the key adds beyond image and identity similarity.

\noindent\textbf{Image-only face analysis.}
Image-only detectors identify manipulation traces through spatial and frequency cues~\citep{Afchar2018MesoNet,Qian2020Frequency,Luo2021Generalizing}, temporal consistency, and cross-modal evidence~\citep{Haliassos2021LipForensics,Zhou2021AVoiD}. Generalization approaches include diffusion-guided disentanglement~\citep{Sun2024DiffusionFake}, multi-generator discrepancy learning~\citep{Yang2025D3}, and multimodal pretraining~\citep{Jung2026SD2}. Related attribution methods identify source generators under closed- or open-world assumptions~\citep{Yang2022DeepFakeAttribution,Zheng2026CAL}. In a publisher workflow, these query-only signals can support review alongside the release-specific evidence supplied by a publisher-held record.

\noindent\textbf{Watermarking and provenance.}
Classical and neural watermarking encode ownership or provenance signals before distribution~\citep{Cox1997Secure,Zhu2018HiDDeN,Tancik2020StegaStamp,Fernandez2023StableSignature}. Robust schemes preserve recoverable evidence through editing~\citep{Wen2023TreeRing,Lu2025VINE,Sander2025WAM,Bui2025TrustMark}. SpecGuard embeds messages through spectral projections in learned features~\citep{Alam2025SpecGuard}. Artificial fingerprinting embeds identifying signals in generator training data for later attribution~\citep{Yu2021Responsible}. Our spectral design relates to phase-domain watermarking~\citep{ORuanaidh1996Phase} and proactive release templates~\citep{Asnani2022Proactive}. Tree-Ring~\citep{Wen2023TreeRing} embeds a Fourier-ring pattern in diffusion latents during generation; by comparison, \method edits an existing image's phase and retains a private record for later comparison.

\noindent\textbf{Proactive deepfake defenses.}
FaceSigns~\citep{Neekhara2024FaceSigns} and SepMark~\citep{Wu2023SepMark} use semi-fragile decoder or tracer outputs. Identity-WM~\citep{zhao2023proactive}, LampMark~\citep{Wang2024LampMark}, and LIDMark~\citep{Wu2026LIDMark} connect embedded signals to identity or landmarks. FaceGuard detects manipulation through disagreement between extracted and expected watermarks~\citep{Yang2021FaceGuard}. MaLP~\citep{asnani2023malp} and FractalForensics~\citep{Wang2025FractalForensics} localize manipulated regions, while multi-face proactive forensics identifies affected faces and identities~\citep{Zhang2026DAWF}. FaceProtect derives dynamic watermarks from facial features~\citep{Lan2024FaceProtect}. RecoverMark~\citep{An2026RecoverMark} embeds recoverable face information in surrounding regions; TFPDD~\citep{Lai2026TFPDD} embeds an identity-aware hash without training. Within this family, \method reads face-relation agreement and release-key retention separately from a private spectral record. Its analytic readings require no fitted model weights, and a learned record model combines their features for joint decisions.

\section{Analysis of Relation and Key Evidence}
\label{sec:score_model}
\label{app:theory_full}
Spectral agreement measures preservation of the enrolled phase relation, while signed projection measures alignment with the release signal. We analyze these complementary roles under explicit noise and dependence assumptions, then compare their scales with the observed wrong-key distribution. Thresholds in the experiments are calibrated empirically.

\subsection{Original-Reference Agreement Concentration}
\label{app:detectability}

For $i=k\in\mathcal{K}_{\rm eff}$, let
$q_i=\cos r_i(\mathbf{J})$ and normalize the amplitude weights as
$a_i=w_i/\sum_j w_j$, so that $H=\sum_i a_iq_i$. If the weighted coefficients were independent,
the natural effective sample size would be
\begin{equation}
n_w
=
\frac{\left(\sum_i w_i\right)^2}{\sum_i w_i^2}
=
\frac{1}{\sum_i a_i^2}.
\label{eq:weighted_ess}
\end{equation}
To include exact dependence within conjugate pairs or larger groups, partition the active bins into blocks $B_b$. Let $a_b=\sum_{i\in B_b}a_i$ and $Q_b=a_b^{-1}\sum_{i\in B_b}a_iq_i$. Then $H=\sum_ba_bQ_b$, $Q_b\in[-1,1]$, and $n_{\rm eff}=1/\sum_ba_b^2\le n_w$.

\begin{proposition}[Separation under independent spectral blocks]
\label{prop:margin}
Let $\mu_c$ be a nominal benign-score mean, $b_a\ge0$ an allowance for a reduction in that mean (e.g., alignment error), and $\mu_d$ an upper bound on the re-rendered mean. Let $\tau$ be the acceptance cutoff on $H$. Assume the block variables $Q_b$ are independent under each hypothesis. If $\mathbb E[H\mid\mathcal H_c]\ge\mu_c-b_a$, $\mathbb E[H\mid\mathcal H_d]\le\mu_d$, and $\mu_d<\tau<\mu_c-b_a$, then
\begin{align}
\Pr(H<\tau\mid\mathcal H_c)&\le\exp[-(\mu_c-b_a-\tau)^2n_{\rm eff}/2],\label{eq:rel_fn}\\
\Pr(H\ge\tau\mid\mathcal H_d)&\le\exp[-(\tau-\mu_d)^2n_{\rm eff}/2].\label{eq:rel_fp}
\end{align}
\end{proposition}
\begin{proof}
Each weighted block $a_bQ_b$ has range length $2a_b$. For any deviation $t>0$, the one-sided bounded-sum inequality gives $\Pr(H-\mathbb E H\ge t)\le\exp[-t^2/(2\sum_ba_b^2)]$, with the same bound for the lower tail. Substituting the two mean-to-threshold gaps yields the result.
\end{proof}
This model explains how weighting and dependence change the available evidence. However, it does not establish independent blocks for real codec residuals or estimate $n_{\rm eff}$ for the evaluated images. We therefore calibrate thresholds empirically and do not use this bound to claim a numerical false-accept rate. Ring cardinality is a support size, not a measure of cryptographic key entropy.

\subsection{Signed Release-Key Projection}
For an ideal unquantized relation-preserving query, write
\[
r_i(\mathbf{J})=\delta_i+\xi_i,
\]
where all sums in this local model run over active bins. Here $i=k$ is a frequency index (not a face-slot index), $\delta_i=\delta_s(k)$ is the enrolled keyed displacement,
$w_i=w_k$, and $\xi_i$ is weakly dependent residual noise with approximately zero mean and uncorrelated with the enrolled displacement. This local model describes the continuous phase response. Since pixel truncation adds signal-dependent changes, verification uses the encoded release's measured projection for normalization.
Substituting into Eq.~\eqref{eq:key_projection} gives
\begin{equation}
G_s(\mathbf{J};R)
=
\frac{\sum_i w_i\delta_i\sin(\delta_i+\xi_i)}
{\sum_i w_i\delta_i^2+\epsilon_0}.
\end{equation}
For small $|\delta_i+\xi_i|$,
\begin{equation}
G_s(\mathbf{J};R)
=
\rho
+
\frac{\sum_i w_i\delta_i\xi_i}
{\sum_i w_i\delta_i^2+\epsilon_0}
+
O\!\left(
\frac{\sum_i w_i|\delta_i|\,|\delta_i+\xi_i|^3}
{\sum_i w_i\delta_i^2+\epsilon_0}
\right),
\end{equation}
where
\begin{equation}
\rho=
\frac{\sum_i w_i\delta_i^2}
{\sum_i w_i\delta_i^2+\epsilon_0}
\approx 1 .
\end{equation}

Here $\rho$ is the attenuation caused by the denominator stabilizer, and $O(\cdot)$ denotes a remainder bounded in order by its argument as the residuals approach zero. Thus the expected correct-key projection stays near its enrolled value under the local noise model; concentration additionally requires control of the weighted noise variance.
For an independently sampled wrong key $s^{-}$, the signed displacement directions are uncorrelated,
so their contributions cancel in expectation:
\begin{equation}
\mathbb{E}\!\left[G_{s^{-}}(\mathbf{J};R)\right]\approx 0,
\qquad
\mathbb{E}\!\left[G_s(\mathbf{J};R)\right]\approx \rho .
\label{eq:key_expect}
\end{equation}
The paired correct-versus-wrong-key experiment tests exactly this distinction by holding the image fixed and changing
only the projection key.

\subsection{Why Ordinary Cosine Is Weakly Key-Sensitive}
\label{app:small_displacement}

Using the same residual model, now assume channel noise is independent of the key direction. Comparing the displaced residual with its undisplaced counterpart gives
\begin{align}
\cos(\xi_i+\delta_i)-\cos\xi_i
&=
-\delta_i\sin\xi_i
-\frac{1}{2}\delta_i^2\cos\xi_i
+O(|\delta_i|^3).
\label{eq:cosine_noise_expansion}
\end{align}
If the residual noise is approximately symmetric about zero, then
$\mathbb{E}[\sin\xi_i]\approx0$, so the first-order contribution cancels in expectation:
\begin{equation}
\mathbb{E}\!\left[
\cos(\xi_i+\delta_i)-\cos\xi_i
\right]
=
-\frac{1}{2}\delta_i^2\,
\mathbb{E}[\cos\xi_i]
+O(|\delta_i|^3).
\label{eq:cosine_second_order}
\end{equation}
For small residual noise, $\mathbb{E}[\cos\xi_i]\approx1$, leaving a key-dependent cosine change
of order $\delta_i^2/2$. Ordinary cosine therefore loses the sign of the weak displacement near
alignment, whereas the signed key projection multiplies by the known displacement direction before
aggregation and retains first-order information.

At the deployed settings, $|\delta_{\rm eff}|$ is approximately $0.0235$ rad on FF++/Celeb-DF and $0.0071$ rad on DFDC, giving second-order cosine terms $\delta_{\rm eff}^2/2$ of $2.76\times10^{-4}$ and $2.52\times10^{-5}$, respectively. This explains the weak local cosine response to the release signal. In the paired wrong-key control, all inputs to $H$ and $C$ stay fixed, so their invariance is exact rather than an approximation at small displacement.

\subsection{Wrong-Key Null}
\label{app:wrong_key_null}

Equation~\eqref{eq:wrong_key_stats} uses $m{=}8$ sampled wrong keys as an empirical null rather than
as a security bound. An independent key produces a signed displacement pattern uncorrelated with
the enrolled one, so its projection averages around zero over the active support. The median
active-bin counts are $19{,}824$, $32{,}461$, and $157{,}684$ on
\textsc{FF++}, \textsc{Celeb-DF}, and \textsc{DFDC}, respectively.

The empirical wrong-key mean is within $2\times10^{-4}$ of zero on each dataset; median per-asset spreads are $0.017$, $0.012$, and $0.006$, respectively. The enrolled lossless query has release-normalized projection $1$ by construction. These sampled distributions do not resolve the far-tail false-accept probability.

\section{Experimental Setup and Reproducibility}
\label{app:experimental_setup}
This section specifies enrollment, model training, query generation, and score conventions used by the supplementary experiments.

\noindent\textbf{Enrollment cohort and crop geometry.}
From $24{,}527$ source frames across \textsc{FF++},
\textsc{DFDC}, and \textsc{Celeb-DF}~\citep{Rossler2019FFpp,Dolhansky2020DFDC,Li2020CelebDF},
faces are located with SCRFD~\citep{Guo2022SCRFD} and accepted only when YuNet or a Haar-cascade face detector also locates an overlapping face box ($24{,}209$ accepted). We center native-resolution crops on the detected face, using width and height ratios $0.224$ and $0.334$ relative to the source frame, and save lossless PNGs without padding or resizing. We partition \textsc{Celeb-DF} by source video to prevent video-level overlap between splits. Each of the $7{,}366$ enrolled real
faces receives its own released image and a private NPZ record storing the enrolled phase, original
phase and amplitude, keyed displacement, active bins, and the release's own key projection.
Record-based verification reads the query and this stored record; it does not reload the original image.

\subsection{Hyperparameters and Training Setup}
\label{app:hyper}

\noindent\textbf{Enrollment and key settings.}
Enrollment writes the keyed displacement at $\alpha=0.05$, so the pre-blend shift is
$0.05\pi$\,rad, with blend weight $\beta=0.15$ on the ring
$\kappa\in[0.10,0.45]$ of normalized frequency
($\alpha=0.015$ for the \textsc{DFDC} release configuration).
Every face passed at the first rung, so these are also the per-face
$(\alpha_i,\beta_i)$ stored in the record, where $i$ indexes the enrolled face.
The key $s$ is the first $64$ bits of SHA-256 over a fixed tag, a master seed, and the sample
identifier; it seeds NumPy's PCG64 generator, whose standard-normal draws $u_s(k)$ over the Fourier
grid give $h=\operatorname{sign}(u_s(k)-u_s(-k))$. Bins with $k=-k$ get $h=0$ and lie outside
$\mathcal{K}_{\rm eff}$. The amplitude floor is $A_{\min}=10^{-6}\max_k A_I(k)$ on the luminance
spectrum of the aligned face; it removes at most $2\%$ of a face's ring bins (median below $0.05\%$).
The wrong-key comparison draws $m{=}8$ keys; feature extraction uses the full ring and $8$ radial bands. The MLP bands use normalized Euclidean distance from the center of the shifted array, divided into eight intervals over the active support. For rectangular crops this differs from the cycles-per-pixel radius $\kappa$ used for enrollment and direct sub-band selection.

\noindent\textbf{Record model and features.}
The feature map $\Phi$ in Eq.~\eqref{eq:record_features} supplies the record model with $75$ inputs: $11$ global
readings and $8$ readings in each of the $8$ radial bands ($11+8\times8=75$). All key projections in this feature list use the release normalization of Sec.~\ref{sec:verifier}; bare $r$ and $\delta$ denote the binwise residual $r_k$ and displacement $\delta_s(k)$. The full-ring block $\mathbf f_0$ contains
$H$, $C$, the amplitude-weighted mean and unweighted standard deviation of the wrapped residual $r$, the weighted mean of
$\cos r$, $G_s$, $\overline{G^{-}}$, $G^{-}_{\max}$, $G_{\Delta}$, $G_z$, and the weighted mean of
$\delta\sin r$. Each block $\mathbf f_\ell$ has $H$, $C$, $G_s$, a band key-distance column, the weighted residual mean and
unweighted standard deviation, and the weighted means of $\cos r$ and $\delta\sin r$. Because the mean of $\cos r$ equals $H$
and the band key-distance column equals that band's $G_s$, only $58$ of the $75$ inputs are distinct.
Each input is standardized with its mean and standard deviation over the training split of the
training dataset (benign, swap, and wrong-key rows), giving $\boldsymbol\mu$ and $\boldsymbol\sigma$ in Eq.~\eqref{eq:record_model_map}. Standard deviations below $10^{-6}$ are replaced by one, and missing feature values are imputed with the corresponding training mean. These statistics and the trained model are frozen during test scoring. The default $\Phi(\mathbf J,R)$ uses the enrolled key; wrong-key training rows substitute the candidate projection direction while keeping the query, original phase, and template fixed.

\par\medskip\noindent\textbf{Training targets.}\phantomsection
\label{app:random_variables}
The record model learns from benign, re-rendered, and wrong-key queries, with the targets in Eq.~\eqref{eq:loss_core} assigned according to query construction. This supervision maps spectral features to relation, key, and manipulation scores for the decision policy. The outputs serve as decision scores rather than calibrated posterior probabilities.

\noindent\textbf{Architecture and optimization.}
The MLP has two 64-unit hidden layers with ReLU activations, three scalar output heads, and two eight-dimensional linear projections for the agreement loss. Sigmoid converts each output logit to its reported score; training evaluates BCE directly on logits for numerical stability.
The record model is trained for $1{,}000$ AdamW steps (learning rate $10^{-3}$, weight decay
$10^{-4}$) on batches of $1{,}024$ rows drawn with replacement, with $\lambda_{\rm rel}=1$,
$\lambda_{\rm key}=0.3$, $\lambda_{\rm agr}=0.1$, and three random seeds, using the objective in Eq.~\eqref{eq:loss_core}.

\noindent\textbf{Agreement-penalty ablation.}
To assess the auxiliary penalty, we train with $\lambda_{\rm agr}=0$ and $0.1$ using matched seeds and FF++ training data. Face-swap AUCs are $0.99972/0.99970$ on FF++, $0.99618/0.99664$ on DFDC, and $0.97385/0.97535$ on Celeb-DF. Removing the penalty changes joint swap acceptance by $0/-0.105/+0.159$ percentage points under target-dataset first-percentile calibration. The penalty has a small, setting-dependent effect. The projections are auxiliary: equal outputs, including zero, minimize their distance independently of the classification scores.

\noindent\textbf{Image-only control and calibration.}
The image-only control uses ResNet-50 at $224^2$. Within-dataset evaluations train on the corresponding dataset; transfer evaluations retain FF++-trained weights. Training uses the main training splits with the following augmentation:
JPEG quality $q=35$--$80$ with application probability $p=0.4$, Gaussian blur with kernel width $3$--$7$ pixels at $p=0.2$, Gaussian noise, and
HSV jitter, all applied at native resolution before resizing. The joint-policy thresholds $\tau_{\rm rel}$ and $\tau_{\rm key}$ use the first percentiles of clean benign training scores, as detailed in App.~\ref{app:joint_policy}. Channel experiments separately calibrate the direct relation and review cutoffs on transformed benign training queries. The threshold-transfer diagnostic in App.~\ref{app:evt_transfer} fits quantiles to test scores separately from the main evaluation thresholds.

\subsection{Evaluation Protocol}
\par\medskip\noindent\textbf{Re-rendering engines.}\phantomsection
\label{app:engines}
Every engine operates on the released image and is scored against that image's own record, and all use
released weights without retraining. SimSwap~\citep{Chen2020SimSwap} supplies the main re-rendered
query; the four re-rendering columns add StarGAN-v2~\citep{Choi2020StarGANv2} (authors' CelebA-HQ EMA checkpoint, $100$k
steps), UniFace~\citep{Xu2022UniFace} (released $500$k-step swap checkpoint), InsightFace \texttt{inswapper\_128} (released
ONNX), and Stable Diffusion v1.5~\citep{Rombach2022StableDiffusion} inpainting over the detected face box ($25$ steps, guidance
$7.5$, fixed seed, no identity donor, so the face is re-rendered by the diffusion prior alone).
The inpainting mask uses a face box detected on the released image with InsightFace's \texttt{buffalo\_l} detector. StarGAN-v2 and SD inpainting also cover all $636$ DFDC and $209$ Celeb-DF test releases with unchanged generation settings; Table~\ref{tab:operating_points} reports their direct verification results. The feature table names $G_{\Delta}$ \texttt{G\_margin}.
For Table~\ref{tab:roundtrip}, UniFace reenactment uses the separate released $1$M-step checkpoint: the enrolled release supplies appearance, while the existing paired attack-source image drives pose and expression. Inputs are resized to $256\times256$. The same $200$ enrollments, delivery channels, and benign-calibrated cutoffs are used for SimSwap, reenactment, and SD inpainting; only SimSwap has clip-level evaluation.

\par\medskip\noindent\textbf{Peer training settings.}
FaceSigns, SepMark, KAD-Net, WAM, and VINE-R use one fixed public encoder--decoder checkpoint per method across all datasets in Table~\ref{tab:peer_detect}. Their first three columns identify test datasets, not training sources. StegaStamp, LIDMark, and our learned controls use dataset-specific weights; transfer fixes source-trained weights for target-dataset evaluation, including the watermark encoder and decoder. Direct template comparison has no trained weights, so its transfer columns repeat the corresponding target-dataset evaluation. Record-model AUCs average three seeds and use only benign and swapped queries; wrong-key rows belong to the separate key-discrimination tests.

\noindent\textbf{Peer-score conventions.}\phantomsection
\label{app:units}
Peer scores have different units: message bit-error rates for FaceSigns, StegaStamp, and KAD-Net; decoder disagreement for SepMark; landmark distance in pixels for LIDMark; and cosine mismatch for our direct test. We therefore compare score order through ROC-AUC, which is invariant to strictly increasing rescaling. Reported signals are oriented toward larger values for suspicious queries, and their means were checked for each method--engine pair. KAD-Net uses its fragile decoder's bit-error rate, whereas SepMark uses its dual-decoder disagreement.

Under StarGAN-v2, LIDMark's mean decoded landmark distance reaches $2.4\times10^{14}$ pixels, compared with $13.1$ on benign queries. Its AUC of $1.000$ reflects separation driven by extreme decoder outputs, not accurate or numerically stable landmark recovery.

\noindent\textbf{Credential retention.} C2PA integration is proposed. The credential-retention experiment uses HMAC-SHA256-authenticated JSON in PNG text or JPEG comment fields to isolate metadata survival; it does not implement a C2PA manifest or validator. When this metadata is stripped, comparison remains possible only if another trusted route resolves the corresponding record.

\section{Release Fidelity and Visual Examples}
\label{app:fidelity_table}

Table~\ref{tab:fidelity} compares release distortion on a shared external pool at three delivery resolutions. Methods use their benchmark configurations with different payload and reference budgets. The following examples show release edits, benign-channel queries, and re-renderings.

Across all $7{,}366$ enrollments, \method records an average PSNR of $51.20$\,dB
($51.22/51.15/51.21$ on \textsc{FF++}/\textsc{DFDC}/\textsc{Celeb-DF}, respectively), with a
minimum of $49.3$\,dB and SSIM rounding to $1.00$ at two decimal places. These values are measured from the released PNG after
writing to and reading back from disk.

\begin{table}[!htbp]
\centering
\footnotesize
\setlength{\tabcolsep}{2pt}\renewcommand{\arraystretch}{1.12}
\caption{\textbf{Release fidelity across datasets and resolutions.} PSNR\,$\uparrow$ / SSIM\,$\uparrow$ on $40$ held-out faces per dataset, resized to $128$/$256$/$512$ pixels. All methods use the same crops and are measured against the corresponding original; payload sizes differ across methods.}
\label{tab:fidelity}
\begin{tabular*}{\linewidth}{@{\extracolsep{\fill}}l ccc ccc ccc@{}}
\toprule
\multicolumn{10}{c}{PSNR (dB) $\uparrow$ / SSIM $\uparrow$} \\
\midrule
& \multicolumn{3}{c}{\textsc{FF++}}
& \multicolumn{3}{c}{\textsc{Celeb-DF}}
& \multicolumn{3}{c}{\textsc{DFDC}} \\
\cmidrule(lr){2-4}
\cmidrule(lr){5-7}
\cmidrule(lr){8-10}
Method
& $128$ & $256$ & $512$
& $128$ & $256$ & $512$
& $128$ & $256$ & $512$ \\
\midrule
StegaStamp
& 29.2/0.98 & 28.8/0.98 & 28.8/0.98
& 29.2/0.98 & 28.9/0.97 & 28.7/0.97
& 29.9/0.98 & 29.4/0.98 & 29.5/0.98 \\

FaceSigns
& 36.2/1.00 & 34.1/1.00 & 35.2/1.00
& 35.3/1.00 & 34.0/1.00 & 34.7/1.00
& 34.6/0.99 & 33.3/0.99 & 33.9/0.99 \\

SepMark
& 41.4/1.00 & 39.7/1.00 & 40.2/1.00
& 41.1/1.00 & 39.7/1.00 & 40.1/1.00
& 41.1/1.00 & 39.4/1.00 & 39.9/1.00 \\

EditGuard
& 41.5/1.00 & 38.4/1.00 & 37.0/1.00
& 38.9/1.00 & 38.1/1.00 & 35.5/1.00
& 40.2/1.00 & 37.4/1.00 & 36.3/1.00 \\

LampMark
& 35.0/1.00 & 32.6/0.99 & 33.7/1.00
& 35.4/1.00 & 34.5/1.00 & 34.9/1.00
& 36.4/1.00 & 33.5/0.99 & 34.6/0.99 \\

KAD-Net
& 43.7/1.00 & 41.1/1.00 & 41.9/1.00
& 43.3/1.00 & 40.9/1.00 & 41.5/1.00
& 43.4/1.00 & 40.8/1.00 & 41.6/1.00 \\

LIDMark
& 42.6/1.00 & 41.3/1.00 & 41.8/1.00
& 36.6/1.00 & 35.2/1.00 & 35.7/1.00
& 42.7/1.00 & 41.2/1.00 & 41.8/1.00 \\

Ours
& 51.2/1.00 & 51.2/1.00 & 51.4/1.00
& 51.3/1.00 & 51.6/1.00 & 51.2/1.00
& 51.1/1.00 & 51.2/1.00 & 51.2/1.00 \\
\bottomrule
\end{tabular*}
\end{table}

\phantomsection
\label{app:supp}
\noindent\textbf{Visual comparisons.}
Fig.~\ref{fig:release_residual} makes the weak release signal visible through amplified residuals. Fig.~\ref{fig:additional_examples} then shows how spectral agreement responds to processing and face replacement across datasets.

\begin{figure}[htbp]
    \centering
    \includegraphics[width=0.82\linewidth]{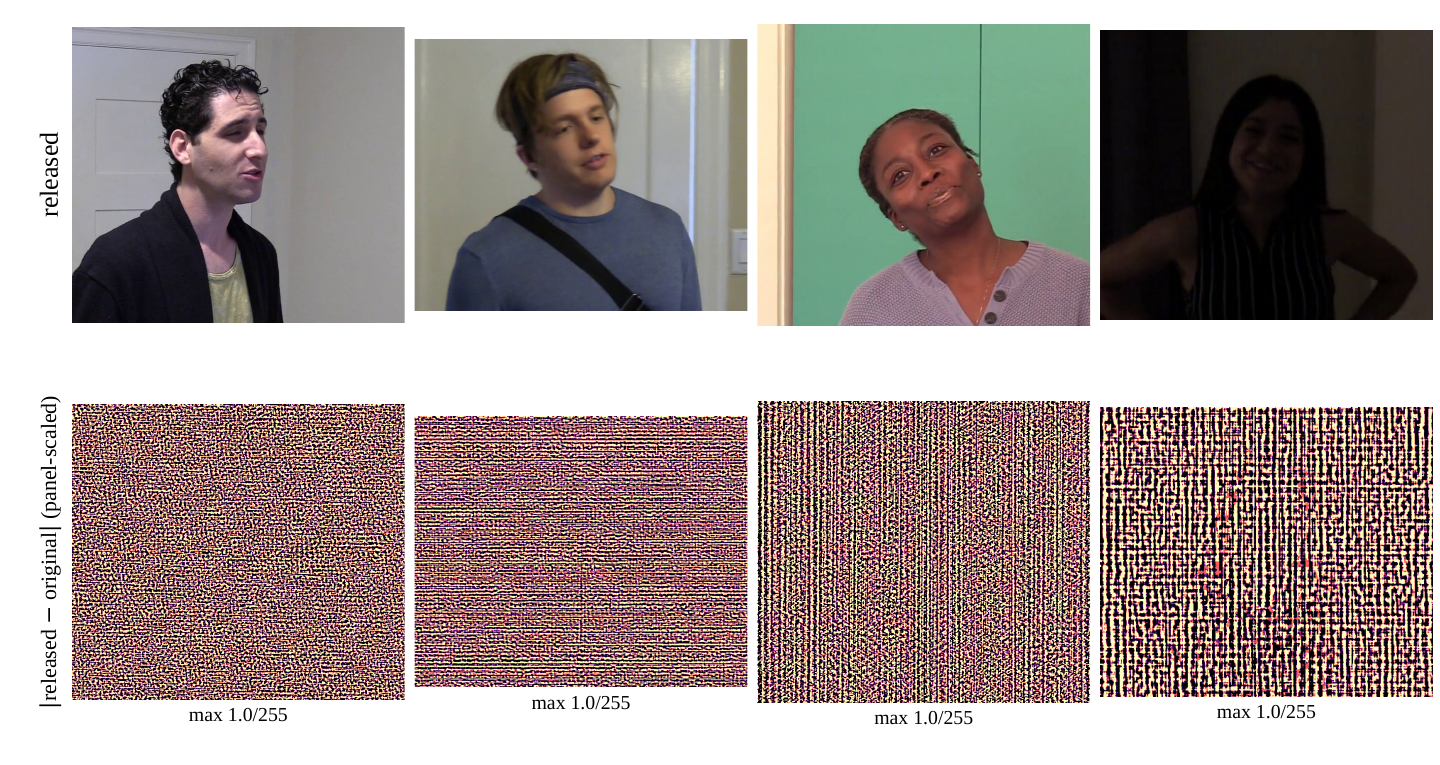}
    \caption{\textbf{Release residuals.} Differences from the original images are scaled separately for visibility. The displayed edits have maximum magnitude $1/255$ and form a spatially diffuse pattern.}
    \label{fig:release_residual}
\end{figure}

\begin{figure}[htbp]
    \centering
    \includegraphics{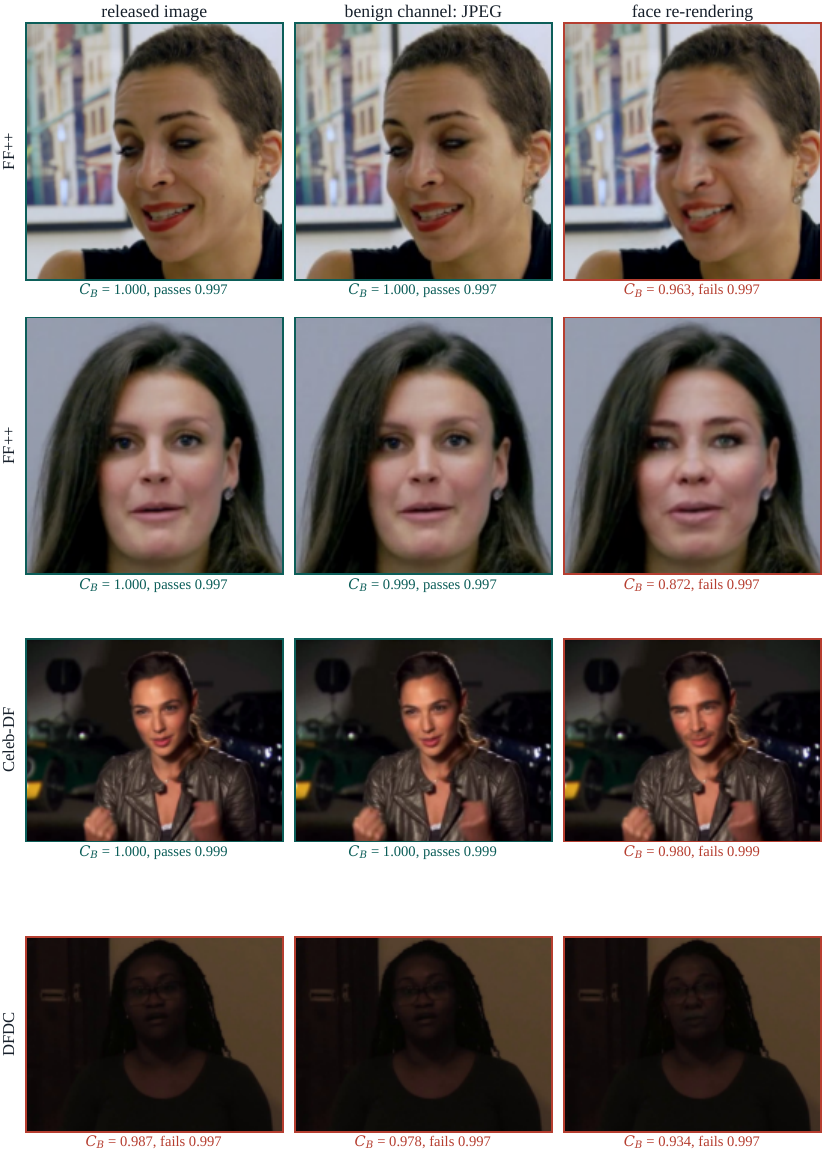}
    \caption{\textbf{Verification examples across datasets.} Released images and JPEG-$q85$ copies retain higher sub-band template agreement $C_B$ than SimSwap queries in these examples. Cutoffs use $1\%$ benign training rejection per dataset; green and red borders indicate relation acceptance and rejection.}
    \label{fig:additional_examples}
\end{figure}

\par\medskip\noindent\textbf{Key pattern and residual.}\phantomsection
\label{app:supp_keymatch}
Fig.~\ref{fig:keymatch} makes the key--template relation visible on one enrolled asset. The key is
a seed that expands into a sign pattern over the ring (a); $(\alpha,\beta)$ only set its size,
producing the keyed displacement $\delta_s$ (b), and the template is the face's own phase carrying that
displacement (c). A benign query's residual against the original phase looks like noise (d),
but its \emph{signs} echo the key's pattern: $82\%$ of active bins agree with the correct key's
signs ($G_s=1.08$, e) against $49\%$, chance, for a wrong key ($G_s=0.01$, f). Here $G_s$ is the unnormalized projection of Eq.~\eqref{eq:key_projection}; panel (f) abbreviates $G_{s^-}$ as $G$ for the wrong key $s^-$. The panel symbols $\delta$ and $r$ denote $\delta_s(k)$ and $r_k(\mathbf J)$ on the displayed bins. The signed projection measures this alignment with the keyed displacement.

\begin{figure}[htbp]
    \centering
    \includegraphics[width=0.92\linewidth]{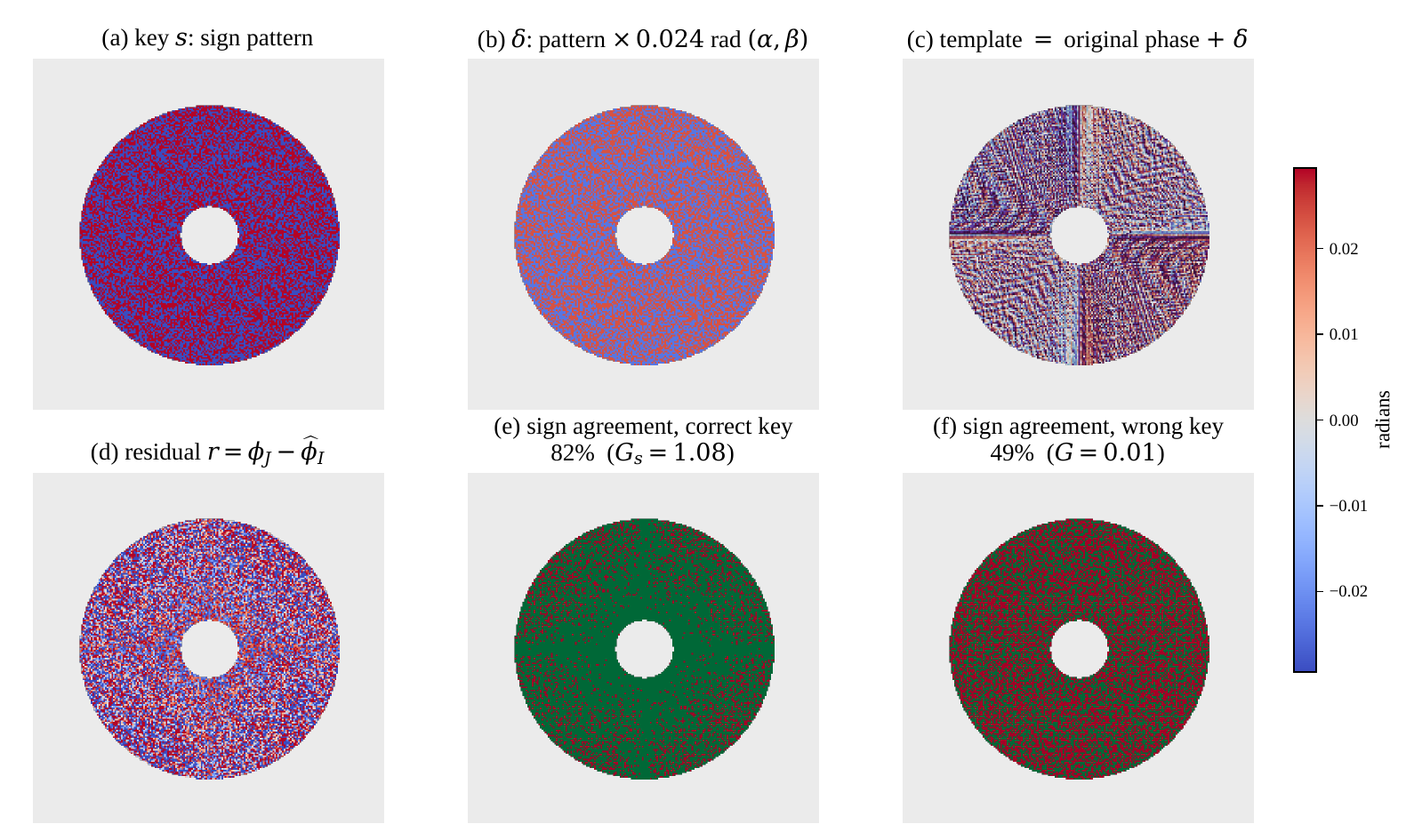}
    \caption{\textbf{Keyed displacement and query residual.} One \textsc{FF++} test asset, showing the central region of the shifted spectrum; gray marks bins outside the active support. Top: the key pattern, displacement, and enrolled template. Bottom: the query residual and its alignment with correct and wrong keys.}
    \label{fig:keymatch}
\end{figure}

\section{Private References and Release-Key Evidence}
\label{app:reference_controls}

Because the direct relation test uses a private per-asset reference, the appropriate controls must
receive the same privilege. The original-phase control reads the undisplaced reference but uses the same keyed query images; it does not evaluate an independently generated unmarked-release cohort. Throughout this section, the key projection is
$G_s^{(\mathrm{template})}$ from Sec.~\ref{sec:verifier}, normalized by the stored template's own
projection rather than by the release projection, so released images retain their natural score
spread.

We compare four reference-based readings on the main test queries: template agreement $C$, original-phase agreement $H$, SSIM against the stored released image, and
ArcFace~\citep{Deng2019ArcFace} similarity to the enrolled face. Swap-detection AUC uses mismatch as the positive direction; release-discrimination AUC instead treats the enrolled release as positive.

\begin{table}[h]
\centering
\footnotesize
\setlength{\tabcolsep}{5pt}
\caption{\textbf{Private-reference comparisons.} Main-test AUC ($\uparrow$) for clean face-swap detection (left), discrimination between releases of the same person (middle), and discrimination between the release and its original image before publication (right).}
\label{tab:reference_controls}
\begin{tabular}{@{}l ccc cc@{}}
\toprule
& \multicolumn{3}{c}{Swap separation (clean) $\uparrow$}
& Same-person release $\uparrow$
& Release-version control $\uparrow$ \\
\cmidrule(lr){2-4}\cmidrule(lr){5-5}\cmidrule(lr){6-6}
Verifier & \textsc{FF++} & \textsc{DFDC} & \textsc{CDF}
& \shortstack{Enrolled release vs.\\another release} & \shortstack{Released image vs.\\unreleased original} \\
\midrule
Template agreement $C$
& 1.000 & 0.978 & 0.984 & -- & -- \\
Original-phase reference
& 1.000 & 0.977 & 0.978 & 1.000 & 0.000 \\
Released-image SSIM
& 1.000 & 1.000 & 1.000 & -- & 1.000 \\
ArcFace identity reference
& 1.000 & 1.000 & 1.000 & 0.9999 & 1.000 \\
Key projection $G_s$
& -- & -- & -- & 1.000 & 1.000 \\
\bottomrule
\end{tabular}
\end{table}

Table~\ref{tab:reference_controls} shows that private references already provide strong clean-swap and same-person release discrimination. To isolate what the key adds beyond these reference capabilities, the original control below tests separation of the released image from its visually matched unreleased original under benign processing.

\paragraph{Unreleased original control.}
This experiment asks whether a verifier can identify the published version while tolerating JPEG compression. For each of $400$ held-out enrollments, we compare the original image, the lossless release, and its JPEG-q90 copy. SSIM and ArcFace use the released image as their reference; spectral tests use its private record. JPEG-q90 thresholds are the first percentiles of all $1{,}144$ FF++, $1{,}580$ DFDC, and $1{,}203$ Celeb-DF benign training releases, computed per dataset and reading. Lossless calibration provides a complementary operating point using $100$ training releases per dataset. Calibration and test source videos are disjoint, and all readings are available. Reported rates pool the dataset-specific decisions.

\begin{table}[h]
\centering\small\setlength{\tabcolsep}{3pt}\renewcommand{\arraystretch}{1.15}
\caption{\textbf{Original and release acceptance under channel calibration.} Acceptance rates on $400$ held-out enrollments, using cutoffs fitted to separate lossless or JPEG-$q90$ training releases at a nominal $1\%$ rejection rate. Lower original acceptance and higher release acceptance are desirable.}
\label{tab:posthoc_attack}
\begin{tabular*}{\linewidth}{@{\extracolsep{\fill}}l ccc ccc@{}}
\toprule
& \multicolumn{3}{c}{Lossless calibration} & \multicolumn{3}{c}{JPEG-q90 calibration} \\
\cmidrule(lr){2-4}\cmidrule(lr){5-7}
Reading & \shortstack{Unreleased\\original} $\downarrow$ & Lossless $\uparrow$ & JPEG $\uparrow$ & \shortstack{Unreleased\\original} $\downarrow$ & Lossless $\uparrow$ & JPEG $\uparrow$ \\
\midrule
Original phase $H$ & 100\% & 89.25\% & 32.75\% & 100\% & 100\% & 99.25\% \\
Template $C$ & 100\% & 89.0\% & 27.0\% & 100\% & 100\% & 99.25\% \\
Released-image SSIM & 0\% & 100\% & 0\% & 100\% & 100\% & 99.5\% \\
ArcFace & 0\% & 98.75\% & 0\% & 99.25\% & 100\% & 99.0\% \\
Key projection $G_s$ & 0\% & 98.5\% & 96.0\% & 0\% & 100\% & 99.5\% \\
\bottomrule
\end{tabular*}
\end{table}

Table~\ref{tab:posthoc_attack} illustrates the tension between compression tolerance and version discrimination. Lossless-calibrated SSIM and ArcFace distinguish the original from the release, but reject JPEG-processed releases. JPEG-calibrated reference tests admit processed releases alongside most originals, while $G_s$ rejects every original and accepts $99.5\%$ of JPEG releases. Thus, the keyed signal supplies version evidence that remains useful when reference similarity must tolerate compression.

Release identity and record age are separate questions: later re-enrollment defines a new valid release. Trusted timestamping establishes when a record existed, while an exact hash binds the pristine file. The keyed displacement supplies a distinct, channel-tolerant measurement of release-specific evidence.

\par\medskip\noindent\textbf{Template specificity and privacy.}\phantomsection
\label{app:template}
Fig.~\ref{fig:template} shows the displaced phase template and its query--record comparisons. Across $24$ enrollments, matched pairs score $0.9974\pm0.0046$, versus $0.103\pm0.091$ for $552$ mismatches; the lowest matched score is $0.979$ and highest mismatch $0.426$ (AUC $1.000$). Despite this specificity, unit-amplitude inversion reveals facial structure, so the template remains sensitive image-derived data. This diagnostic does not isolate the causes of cross-dataset transfer.

\begin{figure*}[t]
\centering
\includegraphics[width=0.92\textwidth]{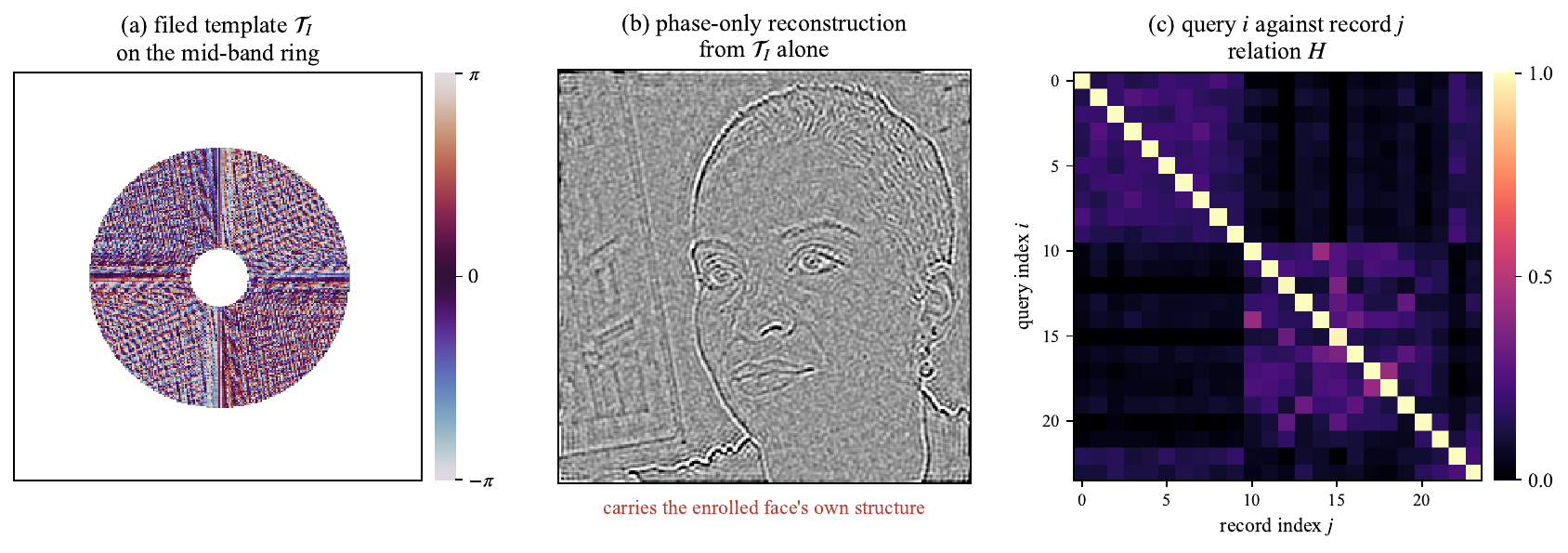}
\caption{\textbf{Enrolled template and asset specificity.} (a) The template $\mathcal T_I$ stores original phase after the keyed displacement on the mid-band ring. (b) Unit-amplitude inversion reveals facial structure, illustrating the template's sensitivity. (c) Original-reference agreement $H$ across $24$ query--record pairs: matched pairs lie on the diagonal and have higher agreement than mismatched pairs.}
\label{fig:template}
\end{figure*}

\par\medskip\noindent\textbf{Paired-key evidence and model scores.}\phantomsection
\label{app:exp_key}
Table~\ref{tab:trigger_signal} holds each query fixed while changing the projection key. The cosine readings are identical within each pair, while key projections separate the paired keys on all three datasets. Mean learned scores distinguish benign, swapped, and wrong-key query types on FF++. Acceptance rates at calibrated cutoffs are reported separately in App.~\ref{app:joint_policy}.

\begin{table}[H]
\centering\footnotesize
\caption{\textbf{Key discrimination and learned record scores.} Top: correct-versus-wrong-key AUC $\uparrow$ with each image held fixed, using $1{,}175$, $636$, and $209$ test pairs on FF++, DFDC, and Celeb-DF; $0.500$ denotes chance. Bottom: mean relation and key scores over three seeds on FF++, with $1{,}175$ queries per type. Bottom arrows mark desired score directions.}
\label{tab:trigger_signal}
\begin{adjustbox}{max width=\linewidth}
\begin{tabular}{@{}llccc@{}}
\toprule
Quantity & Reading & \textsc{FF++} $\uparrow$ & \textsc{DFDC} $\uparrow$ & \textsc{Celeb-DF} v2 $\uparrow$ \\
\midrule
Cosine controls $H/C$ & Spectral agreement & 0.500 / 0.500 & 0.500 / 0.500 & 0.500 / 0.500 \\
Key projections $G_s/G_{\Delta}/G_z$ & keyed displacement & 1.000 / 1.000 / 1.000 & 1.000 / 1.000 / 1.000 & 1.000 / 1.000 / 1.000 \\
\midrule
Query & Decision input & $S_{\rm rel}$ & $S_{\rm key}$ & \\
\midrule
Benign query & both should agree & 0.999 $\uparrow$ & 0.999 $\uparrow$ & \\
Re-rendered after release & both should fail & 0.001 $\downarrow$ & 0.001 $\downarrow$ & \\
Wrong key & key must fail & 1.000 $\uparrow$ & 0.000 $\downarrow$ & \\
\bottomrule
\end{tabular}
\end{adjustbox}
\end{table}

\section{Channel Robustness and Spectral Verification}
\label{app:robust_table}

\par\medskip\noindent\textbf{Fixed enrollment reference.}\phantomsection
\label{app:robust_defect}
Enrollment fixes the template phase, original phase and amplitude, and encoded-release projection used for normalization. Processing acts only on the query. Thus, a lossless release has normalized key projection one, while its template agreement can be below one because the template precedes image quantization. Full-ring $H$ and sub-band $C_B$ both compare queries with this retained reference.

\subsection{Benign Channels Across Datasets}
\label{app:distort_canonical}

Table~\ref{tab:benign_perds} compares full-ring agreement with dataset-trained image-only detectors. Spectral separation varies with geometry, active-bin counts, and release strength; median crops are $256{\times}256$ on FF++, $379{\times}284$ on Celeb-DF, and $759{\times}685$ on DFDC. Distortion-augmented training strengthens image-only detection under several blur and codec conditions. The analytic sub-band alternative instead adjusts the frequencies used for comparison, as examined below.

\begin{table}[ht]
\centering
\small\setlength{\tabcolsep}{3pt}\renewcommand{\arraystretch}{1.15}

\caption{\textbf{Benign-channel robustness across datasets} (AUC $\uparrow$). Main-test AUCs averaged within each channel family. The $1-H$ rows use full-ring original-reference agreement; image-only detectors are trained on their respective datasets with distortion augmentation. Table~\ref{tab:robust_detail_main} reports the sub-band comparison.}
\label{tab:benign_perds}
\begin{tabular*}{\linewidth}{@{\extracolsep{\fill}}ll cccc c@{}}
\toprule
Dataset & Reading & JPEG $\uparrow$ & Blur $\uparrow$ & Codec $\uparrow$ & Pixel $\uparrow$ & Clean $\uparrow$ \\
\midrule
\multirow{2}{*}{\textsc{FF++}}
& $1-H$ (full ring) & 0.994 & 0.597 & 0.756 & 0.996 & 1.000 \\
& \dfh{}             & 0.972 & 0.773 & 0.776 & 0.976 & 0.986 \\
\addlinespace[1.5pt]

\multirow{2}{*}{\textsc{DFDC}}
& $1-H$ (full ring) & 0.799 & 0.568 & 0.589 & 0.879 & 0.977 \\
& \dfh{}             & 0.995 & 0.974 & 0.910 & 0.997 & 0.997 \\
\addlinespace[1.5pt]

\multirow{2}{*}{\textsc{Celeb-DF}}
& $1-H$ (full ring) & 0.736 & 0.533 & 0.550 & 0.809 & 0.978 \\
& \dfh{}             & 0.967 & 0.896 & 0.831 & 0.969 & 0.979 \\
\bottomrule
\end{tabular*}
\end{table}

These AUCs measure face-swap separation. Table~\ref{tab:trigger_signal} and App.~\ref{app:reference_controls} examine release-key evidence; App.~\ref{app:joint_policy} combines both forms of evidence in joint decisions.

\subsection{Preserved Frequencies and Verification Performance}
\label{app:readout_ceiling}
\label{app:readout}

\noindent\textbf{Reading preserved evidence.}
Stored original-image amplitudes emphasize strong enrolled coefficients, supporting separation under JPEG compression and pixelation. Blur and transcoding can still suppress initially strong upper-ring coefficients. A lower sub-band concentrates comparison on frequencies more likely to preserve the enrolled relation.

To evaluate the fixed $\kappa\le0.15$ setting on further sources, we select $100$ additional FF++ videos, excluding sources in the recorded experiment manifests and the pilot evaluation. Two predetermined frames per video yield $197$ valid release--SimSwap pairs after three face-detection failures, covering all $100$ videos; replacement identities come from training data. The cohort, generation settings, frequency cut, and thresholds are fixed before evaluation.

\noindent\textbf{Separation and operating points.}
Table~\ref{tab:fixed_band_sources} compares $C$ and $C_B$ on these paired queries. Both readings use first-percentile agreement thresholds from the same $1{,}144$ benign training releases, calibrated per channel. JPEG uses quality $85$; blur uses Gaussian radius $3$; reposting applies $0.75\times$ bicubic resizing and JPEG quality $70$ before restoring enrollment geometry; H.264 encodes each image independently at CRF $30$.

\begin{table}[ht]
\centering\small\setlength{\tabcolsep}{4pt}
\caption{\textbf{Face-swap detection after image processing.} Restricting spectral comparison to the fixed lower sub-band improves swap rejection under blur and reposting. After H.264 encoding, higher AUC does not translate into better rejection at the calibrated threshold. TPR and FPR report the percentages of swapped and benign queries rejected, respectively.}
\label{tab:fixed_band_sources}
\begin{tabular*}{\linewidth}{@{\extracolsep{\fill}}l rrr rrr@{}}
\toprule
& \multicolumn{3}{c}{Full ring $C$} & \multicolumn{3}{c}{Sub-band $C_B$} \\
\cmidrule(lr){2-4}\cmidrule(lr){5-7}
Channel & AUC $\uparrow$ & TPR $\uparrow$ & FPR $\downarrow$ & AUC $\uparrow$ & TPR $\uparrow$ & FPR $\downarrow$ \\
\midrule
Clean & 1.000 & 100.0 & 0.0 & 1.000 & 100.0 & 0.0 \\
JPEG $q85$ & 0.998 & 88.8 & 0.0 & 1.000 & 100.0 & 0.0 \\
Blur & 0.592 & 21.3 & 12.2 & 0.922 & 60.9 & 5.6 \\
Social repost & 0.971 & 51.3 & 0.0 & 1.000 & 92.9 & 0.0 \\
H.264 intra & 0.849 & 6.6 & 0.0 & 0.973 & 3.6 & 0.0 \\
\bottomrule
\end{tabular*}
\end{table}

Under blur, the sub-band raises AUC by $0.330$; $2{,}000$ paired source-video bootstrap resamples give a $95\%$ interval of $[0.299,0.353]$, conditional on the fixed training calibration. Together with improved reposting rejection, this supports frequency restriction for preserving useful face evidence. Operating-point transfer remains channel dependent: blur FPR is $5.6\%$ despite the $1\%$ training target, and H.264 AUC improves without higher TPR. Channel selection therefore depends on both benign acceptance and swap rejection at the deployed threshold.

\noindent\textbf{Relation and release evidence.}
Strong processing can push a released image's key score toward that of its unreleased original while retaining useful relation separation. Thus, frequency restriction alone does not establish version discrimination. The JPEG controls in Table~\ref{tab:reference_controls_main} and joint channel tests in App.~\ref{app:joint_policy} evaluate whether both forms of evidence remain usable.

\noindent\textbf{Cropping as a known mask.}
The verifier reconstructs a spatial reference from stored amplitudes and template phases, applies the known crop mask, and compares Fourier phase with the mean-centered query, zero-padded to enrollment geometry. Comparison uses amplitude weights over the full support and a $\pm2$-pixel placement search, calibrated on $1{,}144$ benign training crops. A center crop retaining $90\%$ of each dimension yields AUC $0.962$ over $200$ pairs, accepting $99\%$ of benign queries but rejecting only $6\%$ of swaps. Geometry correction therefore restores benign agreement more readily than reliable swap rejection.

\noindent\textbf{Clip-level aggregation.}
Averaging $C_B$ within an enrolled video reduces the influence of individual degraded frames. Calibration uses the first percentile of mean scores from $120$ benign training videos. On $120$ test videos, averaging accepts $98.3\%$ of benign clips and rejects $65\%$ of swaps, with AUC $0.986$, versus $5.7\%$ swap rejection per frame. This supports temporal aggregation for independently encoded images; inter-frame coding remains to be evaluated.

\section{Operating Points and Calibration}
\label{app:exp_detail}
We distinguish descriptive test-ROC operating points from training-calibrated decisions, then examine reuse of thresholds across datasets.

\subsection{Empirical Detection Operating Points}
\label{app:alarm_operating}
\label{app:operating_points}

To resolve differences hidden by near-ceiling AUC, Table~\ref{tab:operating_points} reports recall along the empirical test ROC at specified false-positive limits. At empirical FPR at most $10^{-3}$, all four additional engines reach recall $1.000$ over $1{,}175$ FF++ pairs each. StarGAN-v2 and SD inpainting also reach this operating point on all $636$ DFDC and $209$ Celeb-DF pairs, extending the observed separation to GAN- and diffusion-based manipulations across datasets. SimSwap recall is $0.999$ on \textsc{FF++} and $0.734$ on \textsc{DFDC}, despite AUCs of $1.000$ and $0.992$. These are descriptive test-ROC operating points, not results at independently calibrated deployment thresholds.

\begin{table}[h]
\centering\footnotesize\setlength{\tabcolsep}{4pt}
\caption{\textbf{Empirical operating points of the direct relation test.} Template mismatch $1-C_B$; $n$ counts enrolled pairs, with benign queries as negatives. Each TPR is the largest test-ROC value at or below the target empirical FPR. Targets below $1/n$ use the zero-false-positive segment. Thresholds are read from the test ROC rather than calibrated on training data.}
\label{tab:operating_points}
\begin{tabular*}{\linewidth}{@{\extracolsep{\fill}}ll r ccc@{}}
\toprule
Re-rendering & Dataset & $n$ & AUC $\uparrow$ & TPR@FPR $10^{-2}$ $\uparrow$ & TPR@FPR $10^{-3}$ $\uparrow$ \\
\midrule
StarGAN-v2 & \textsc{FF++} & 1{,}175 & 1.000 & 1.000 & 1.000 \\
StarGAN-v2 & \textsc{DFDC} & 636 & 1.000 & 1.000 & 1.000 \\
StarGAN-v2 & \textsc{Celeb-DF} & 209 & 1.000 & 1.000 & 1.000 \\
UniFace & \textsc{FF++} & 1{,}175 & 1.000 & 1.000 & 1.000 \\
InSwapper & \textsc{FF++} & 1{,}175 & 1.000 & 1.000 & 1.000 \\
SD inpainting & \textsc{FF++} & 1{,}175 & 1.000 & 1.000 & 1.000 \\
SD inpainting & \textsc{DFDC} & 636 & 1.000 & 1.000 & 1.000 \\
SD inpainting & \textsc{Celeb-DF} & 209 & 1.000 & 1.000 & 1.000 \\
SimSwap & \textsc{FF++} & 1{,}175 & 1.000 & 0.999 & 0.999 \\
SimSwap & \textsc{Celeb-DF} & 209 & 0.991 & 0.938 & 0.928 \\
SimSwap & \textsc{DFDC} & 636 & 0.992 & 0.868 & 0.734 \\
\bottomrule
\end{tabular*}
\end{table}

\subsection{Joint acceptance with training-only calibration}
\label{app:joint_policy}
Table~\ref{tab:joint_policy} evaluates the confirmed branch of Eq.~\eqref{eq:accept} on clean queries using three fixed FF++-trained record models. Calibration uses 1,144, 1,580, and 1,203 benign training faces on FF++, DFDC, and Celeb-DF, with source-video groups disjoint from testing. The original per-head first-percentile thresholds yield joint calibration rejection of $1.57\%$, $1.88\%$, and $1.88\%$, respectively; the matched comparison below instead controls their conjunction.

\noindent\textbf{Direct versus learned decisions.}
To compare the decision paths at a common budget, we evaluate full-ring $C$ and release-normalized $G_s$ against learned $S_{\rm rel},S_{\rm key}$ on the same queries. For each method, both thresholds use a common lower-tail quantile, chosen as the largest value between zero and $0.01$ for which joint rejection on benign training queries does not exceed $1\%$. This matches the calibration budget rather than each head's rejection rate; test queries do not select thresholds. Realized calibration rejection is $0.962\%$, $0.949\%$, and $0.998\%$ on FF++, DFDC, and Celeb-DF for both methods. Test benign rejection can differ. To assess joint decisions after delivery, we use the $200$ FF++ test pairs in Table~\ref{tab:roundtrip} and calibrate each policy on $100$ benign training releases per channel at $1\%$ joint rejection. This comparison covers benign and swapped queries with fixed learned weights. JPEG uses quality $85$; reposting applies $0.75\times$ bicubic resizing and JPEG quality $70$, followed by restoration to enrollment geometry.

\begin{table}[ht]
\centering\small\setlength{\tabcolsep}{3pt}\renewcommand{\arraystretch}{1.15}
\caption{\textbf{Joint verification with per-head and joint calibration budgets.} Under a matched joint budget, direct thresholds match FF++ swap rejection and reject more transferred swaps; the learned policy slightly increases benign acceptance on FF++ and DFDC. Rates are percentages; learned results are means over three FF++ seeds.}
\label{tab:joint_policy}
\begin{tabular*}{\linewidth}{@{\extracolsep{\fill}}llrrrr@{}}
\toprule
\shortstack{Dataset /\\channel} & Policy & $n$ & \shortstack{Benign\\confirmed $\uparrow$} & \shortstack{SimSwap\\confirmed $\downarrow$} & \shortstack{Wrong key\\confirmed $\downarrow$} \\
\midrule
\multicolumn{6}{@{}l}{\textbf{Original policy: $1\%$ calibration rejection per head}}\\
FF++ & Learned & 1,175 & $99.15\pm0.34$ & $0.09\pm0.00$ & $0.00\pm0.00$ \\
DFDC & Learned & 636 & $97.33\pm0.16$ & $2.04\pm0.00$ & $0.00\pm0.00$ \\
Celeb-DF & Learned & 209 & $100.00\pm0.00$ & $8.93\pm0.28$ & $0.00\pm0.00$ \\
\addlinespace[4pt]
\multicolumn{6}{@{}l}{\textbf{Matched comparison: at most $1\%$ joint calibration rejection}}\\
FF++ & Direct $C,G_s$ & 1,175 & 99.57 & 0.09 & 0.00 \\
& Learned & 1,175 & $99.69\pm0.26$ & $0.09\pm0.00$ & $0.00\pm0.00$ \\
DFDC & Direct $C,G_s$ & 636 & 97.64 & 0.79 & 0.00 \\
& Learned & 636 & $98.01\pm0.36$ & $2.31\pm0.18$ & $0.00\pm0.00$ \\
Celeb-DF & Direct $C,G_s$ & 209 & 100.00 & 3.83 & 0.00 \\
& Learned & 209 & $100.00\pm0.00$ & $14.04\pm0.73$ & $0.00\pm0.00$ \\
\addlinespace[4pt]
\multicolumn{6}{@{}l}{\textbf{FF++ delivery: $1\%$ joint calibration rejection per channel}}\\
Lossless & Direct $C,G_s$ & 200 & 99.00 & 0.00 & --- \\
& Learned & 200 & $95.00\pm0.00$ & $0.00\pm0.00$ & --- \\
JPEG $q85$ & Direct $C,G_s$ & 200 & 99.00 & 2.00 & --- \\
& Learned & 200 & $99.83\pm0.29$ & $9.17\pm0.29$ & --- \\
Repost & Direct $C,G_s$ & 200 & 97.50 & 15.50 & --- \\
& Learned & 200 & $99.50\pm0.00$ & $38.83\pm1.76$ & --- \\
\bottomrule
\end{tabular*}
\end{table}

\Needspace{7\baselineskip}
\noindent\textbf{Implication for the verifier.}
On these clean queries, direct thresholds provide competitive joint rejection without training, while the record model trades slightly higher benign acceptance for greater transferred-swap acceptance. Both reject every tested wrong-key query. Direct verification therefore provides the starting point when swap rejection is the priority. A similar tradeoff appears under JPEG and reposting: the learned policy retains more processed releases but also accepts more swaps. Both policies lose swap rejection after reposting, reinforcing the need to evaluate joint relation--key verification for each supported channel. These comparisons use full-ring direct scores and do not establish robustness to adaptive attacks.

\noindent\textbf{Other manipulations and operating-point sensitivity.}
We evaluate StarGAN-v2, UniFace, InSwapper, and SD inpainting on the same $1{,}175$ FF++ release pairs per method, using the $1{,}144$ benign training queries for calibration. At a $1\%$ joint calibration budget, both policies reject all four types of manipulation. Benign test rejection is $0.43\%$ for direct thresholds and $0.31\%$ for the learned policy, averaged over three seeds. Extending StarGAN-v2 and SD inpainting to all $636$ DFDC and $209$ Celeb-DF releases, both policies again reject every generated query using the existing target-dataset thresholds and fixed FF++ model weights. Benign rejection is $2.36\%$ for direct verification and $1.99\%$ for the learned policy on DFDC, and zero observed for both on Celeb-DF. These clean generative manipulations therefore support cross-dataset verification without showing a learned advantage in manipulation rejection.

A sweep of $101$ joint calibration budgets from $0\%$ to $10\%$ reveals an operating-point tradeoff under reposting. At a $10\%$ training budget, the learned policy reaches test TPR $93.83\%$ and FPR $2.33\%$; the direct policy at its $1\%$ budget reaches $84.50\%$ and $2.50\%$, respectively. This exploratory comparison shows that the learned path can be useful at other operating points, but deployment budget selection requires separate validation.

\subsection{Cross-Dataset Threshold Transfer}
\label{app:evt_transfer}

To distinguish weight transfer from threshold transfer, Table~\ref{tab:evt_transfer} holds the FF++-trained record model fixed and varies the calibration dataset. An empirical quantile fitted to one dataset's benign test scores is applied unchanged to the other datasets. Off-diagonal rejection rates range from zero observed to approximately $125$ times the target. Dataset and enrollment strength vary together for DFDC; the matrix does not isolate their effects or measure manipulation recall. These diagnostic test-set quantiles are separate from the training-only calibration used for the joint policy.

\begin{table}[H]
\centering\small\setlength{\tabcolsep}{3pt}\renewcommand{\arraystretch}{1.15}
\caption{\textbf{Cross-dataset threshold transfer.} Benign rejection rates from the fixed FF++ record model using $S_{\rm rel}$. Each row fits a test-score quantile on the marked dataset ($\checkmark$) and applies it unchanged to the others. Scores pool three seeds over $1{,}175$, $209$, and $636$ distinct benign faces, respectively. Entries measure empirical rejection, subject to finite-sample resolution; zeros denote no observed rejections.}
\label{tab:evt_transfer}
\begin{tabular*}{\linewidth}{@{\extracolsep{\fill}}ll ccc@{}}
\toprule
Fitted on & Target FPR & \multicolumn{3}{c}{FPR realized when reused on $\downarrow$} \\
\cmidrule(lr){3-5}
 & & \textsc{FF++} & \textsc{Celeb-DF} & \textsc{DFDC} \\
\midrule
\multirow{2}{*}{\textsc{FF++}}
 & $10^{-2}$ & \best{$1.0\!\times\!10^{-2}$\,$\checkmark$} & $0$ & $9.7\!\times\!10^{-2}$ \\
 & $10^{-3}$ & \best{$1.1\!\times\!10^{-3}$\,$\checkmark$} & $0$ & $6.4\!\times\!10^{-2}$ \\
\addlinespace[2pt]
\multirow{2}{*}{\textsc{Celeb-DF}}
 & $10^{-2}$ & $7.0\!\times\!10^{-2}$ & \best{$1.1\!\times\!10^{-2}$\,$\checkmark$} & $1.5\!\times\!10^{-1}$ \\
 & $10^{-3}$ & $3.9\!\times\!10^{-2}$ & \best{$1.6\!\times\!10^{-3}$\,$\checkmark$} & $1.3\!\times\!10^{-1}$ \\
\addlinespace[2pt]
\multirow{2}{*}{\textsc{DFDC}}
 & $10^{-2}$ & $0$ & $0$ & \best{$1.0\!\times\!10^{-2}$\,$\checkmark$} \\
 & $10^{-3}$ & $0$ & $0$ & \best{$1.0\!\times\!10^{-3}$\,$\checkmark$} \\
\bottomrule
\end{tabular*}
\end{table}

\subsection{Calibration Sensitivity and Finite-Sample Resolution}
\label{app:tail_resolution}
Fig.~\ref{fig:hyper}(a--b) separates the target rejection rate from the amount of calibration data. Relation recall remains $0.999$ across the tested targets, while increasing the calibration cohort from $25$ to $100$ faces reduces variability in benign rejection. Panel (c) reports the record model's key-loss sensitivity.

With $n$ independent negative faces, empirical FPR changes in increments of $1/n$; zero observed errors does not establish a lower population error rate. The paired-key tests contain 1,175, 636, and 209 negative faces, and repeated training seeds do not increase these counts. All joint-policy cutoffs use empirical training quantiles without tail extrapolation. Because either failed criterion prevents confirmation, a joint rejection budget must be calibrated on the combined decision, as in the matched comparison above.

\begin{figure}[htbp]
\centering

\includegraphics[width=0.92\linewidth]{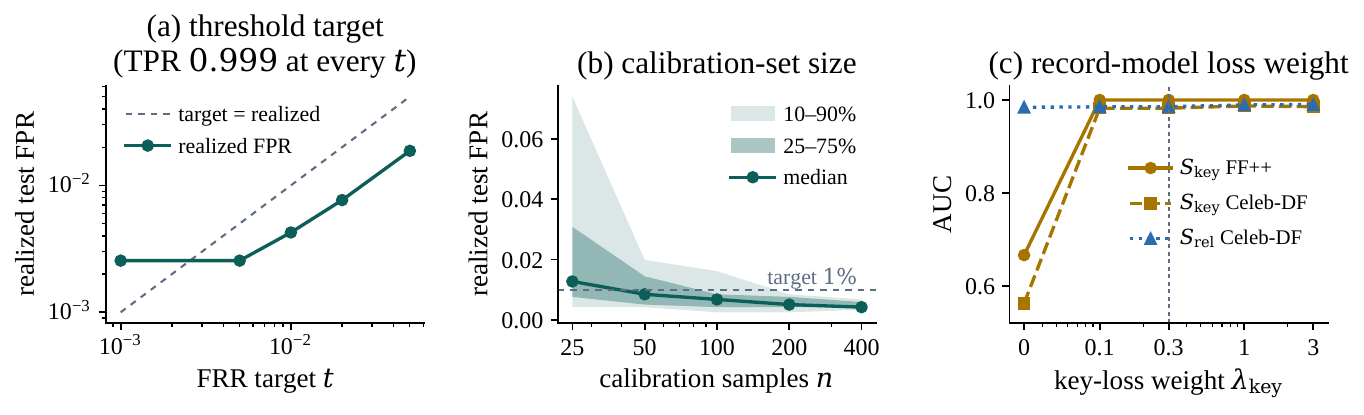}
\caption{\textbf{Calibration and record-model sensitivity} (\textsc{FF++} unless noted).
(a) Target benign false-rejection rate (FRR), denoted $t$ in this panel; (b) number $n$ of calibration samples, with shaded percentile ranges; and (c) key-loss weight $\lambda_{\rm key}$.
}
\label{fig:hyper}

\end{figure}

\section{Sensitivity of the Release Mechanism}
\label{app:sensitivity}

To isolate the release mechanism from learned compensation, the strength and ring sweeps use analytic readings with no fitted model weights. The key-loss sweep separately evaluates the record model. For every
thresholded quantity, the operating point is recalibrated independently at each setting using benign-training quantiles ($n{=}200$ training assets per dataset and setting). The same
recorded per-asset SimSwap replacement identities are reused throughout so that settings remain directly comparable.

Fig.~\ref{fig:sensitivity} shows stable relation and wrong-key separation over the tested strength and ring ranges, with positive gaps between JPEG-processed releases and unreleased originals. Ring effects differ across datasets, where crop geometry and strength also vary. Table~\ref{tab:annulus_sens} details the FF++ results; Fig.~\ref{fig:sensitivity_m} measures the cost of additional wrong keys.

\begin{figure}[p]
\centering
\includegraphics{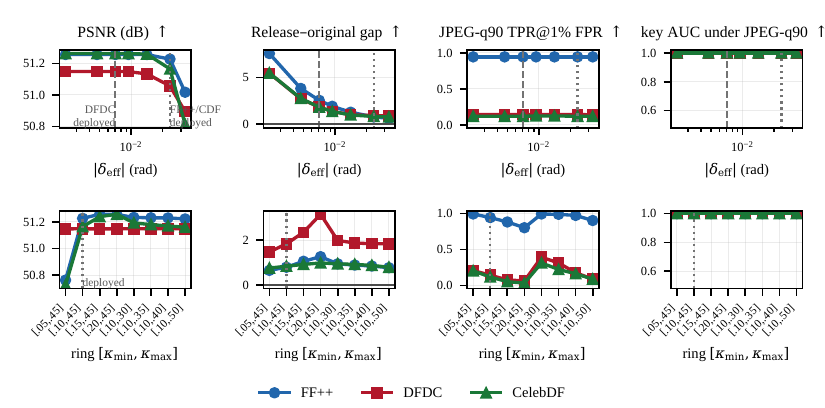}
\captionof{figure}{\textbf{Effects of enrollment strength and frequency-ring location.}
Varying the embedding strength (top) and ring location (bottom) reveals their effects on visual fidelity, release--original separation, and verification performance. Columns show PSNR, the key-projection gap between releases and originals, relation TPR after JPEG-q90, and correct-vs-wrong-key AUC. TPR uses thresholds calibrated to $1\%$ benign training rejection. Dashed and dotted lines indicate the settings used in our experiments.}
\label{fig:sensitivity}

\par\vspace{10pt}

\includegraphics{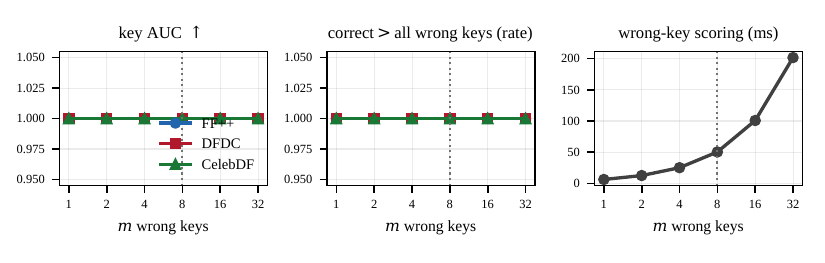}
\captionof{figure}{\textbf{Wrong-key count.}
Key separation remains stable from $m{=}1$ to $32$, while scoring cost grows approximately
linearly. The evaluated configuration uses $m{=}8$.}
\label{fig:sensitivity_m}

\end{figure}

\begin{table}[htbp]
\centering
\footnotesize
\setlength{\tabcolsep}{3pt}
\renewcommand{\arraystretch}{1.15}
\caption{\textbf{Effect of ring location on visual quality and verification} (FF++). The tested rings preserve high visual fidelity and clean face-swap separation, while JPEG robustness and the release--original key-projection gap vary with ring location. Ring bounds $[\kappa_{\min},\kappa_{\max}]$ are the inner and outer frequency radii, normalized so that zero is DC and one is the Fourier-plane corner. Strength is fixed at the deployed value, with calibration repeated for each setting.}
\label{tab:annulus_sens}
\begin{tabular*}{\linewidth}{@{\extracolsep{\fill}}l c cc cc@{}}
\toprule
& & \multicolumn{2}{c}{Face-swap detection AUC $\uparrow$} & \multicolumn{2}{c}{Key evidence} \\
\cmidrule(lr){3-4}\cmidrule(lr){5-6}
\shortstack{Ring\\$[\kappa_{\min},\kappa_{\max}]$} & \shortstack{PSNR\\(dB) $\uparrow$} & Clean & JPEG-q90 & \shortstack{Release--original\\key-projection gap $\uparrow$} & \shortstack{Correct vs. wrong key\\AUC $\uparrow$} \\
\midrule
$[0.05,\,0.45]$ & 50.8 & 1.000 & 1.000 & 0.65 & 1.000 \\
$[0.10,\,0.45]$ & 51.2 & 1.000 & 0.998 & 0.79 & 1.000 \\
$[0.15,\,0.45]$ & 51.3 & 1.000 & 0.995 & 1.06 & 1.000 \\
$[0.20,\,0.45]$ & 51.3 & 1.000 & 0.990 & 1.26 & 1.000 \\
$[0.10,\,0.30]$ & 51.2 & 1.000 & 1.000 & 0.94 & 1.000 \\
$[0.10,\,0.35]$ & 51.2 & 1.000 & 0.999 & 0.89 & 1.000 \\
$[0.10,\,0.40]$ & 51.2 & 1.000 & 0.999 & 0.84 & 1.000 \\
$[0.10,\,0.50]$ & 51.2 & 1.000 & 0.997 & 0.77 & 1.000 \\
\bottomrule
\end{tabular*}
\end{table}

\noindent\textbf{Key-projection gap between releases and originals.}
This quantity subtracts the highest key projection among unreleased originals from the first percentile of projections among JPEG-q90 releases. Thus, a positive gap means that at least $99\%$ of the evaluated releases score above every evaluated original, leaving room for a threshold that separates them. It measures separation from unreleased originals, whereas the final AUC column measures separation from wrong keys.

For the record model, removing the key loss reduces key separation toward chance; positive weights from $0.1$ to $3$ yield similar separation while leaving the relation head largely unchanged. Fig.~\ref{fig:hyper}(c) reports this diagnostic.

\section{Deployment Interface and Cost}\label{app:deploy}

The publisher retains the spectral record and the metadata needed to associate a query with its enrolled face. We measure the cost of enrollment and verification with that association already resolved; App.~\ref{app:limits} discusses operational conditions and record protection.

\par\medskip\noindent\textbf{Runtime and storage.}\phantomsection
\label{app:deploy_cost}
We benchmark the accompanying implementation on 20 native $256^2$ FF++ test face crops, after three warm-up iterations, using an Intel Xeon Gold 6248R CPU with one PyTorch thread. Enrollment, including the RGB edit, uint8 conversion, and nine release-normalization scalars, takes $66.7$\,ms per crop; this excludes quality-gate checks and file serialization. Direct $C_B$ comparison with an in-memory record takes $5.9$\,ms. Loading the query and NPZ, extracting all 75 features with eight alternative keys, and running the approximately $10.3$k-parameter MLP takes $149.9$\,ms. These stages exclude face detection, alignment, and record lookup.

Table~\ref{tab:deploy_cost} reports the CPU verification workloads alongside peer measurements on an A5000. Hardware and processing paths are specified separately for each timing.

Storage depends on the retained support and numerical precision. On 150 FF++ faces, the codec experiment serializes phase, amplitude, and support-mask payloads: float16 arrays require 69.6\,KiB, 8-bit phase and 4-bit log-amplitude require 28.0\,KiB, and restricting the quantized payload to $\kappa\le0.15$ requires 2.3\,KiB. These payload variants preserve detection and correct-versus-wrong-key AUCs to three decimals on the tested assets; this codec comparison measures ranking performance.

These codec payloads exclude dimensions, keys, strengths, geometry, calibration metadata, and nine float64 release projections (72 bytes). The sub-band codec retains evidence for restricted comparison, whereas the model needs full-ring features. Runtime measurements use the $356$\,KiB research NPZ or its in-memory arrays; compact-codec loading and direct full-ring joint latency remain unmeasured.

The optional image-only review signal has $91$\,MiB of weights and used $185$\,MiB of GPU memory with $6.9$\,ms inference in the A5000 benchmark. Its cost is separate from the CPU record verifier.

\begin{table}[tb]
\centering\footnotesize\setlength{\tabcolsep}{3pt}\renewcommand{\arraystretch}{1.08}
\caption{\textbf{Runtime and records for the evaluated verification paths.} Means over 20 $256^2$ crops after warm-up. Peer timings use an A5000; our timings ($^\ddagger$) use a Xeon Gold 6248R CPU and the listed records. CPU timings cover relation-only comparison or the learned path, excluding the image-only detector. Compact codecs are evaluated separately.}
\label{tab:deploy_cost}
\begin{tabular*}{\linewidth}{@{\extracolsep{\fill}}l rrrr l@{}}
\toprule
& \multicolumn{2}{c}{Time (ms) $\downarrow$} & \multicolumn{2}{c}{Memory (MiB) $\downarrow$} & Record used \\
\cmidrule(lr){2-3}\cmidrule(lr){4-5}
Method & Embed & Verify & Weights & GPU & Spectral record \\
\midrule
StegaStamp   & 9.3 & 2.6 & 208 & 375 & none \\
FaceSigns   & 4.4 & 1.7 & 72 & 102 & none \\
SepMark             & 12.5 & 17.8 & 550 & 653 & none \\
WAM                 & 14.9 & 17.0 & 360 & 450 & none \\
KAD-Net              & 17.0 & 12.5 & 563 & 650 & none \\
VINE-R                 & 79.3 & 10.7 & 3833 & 4569 & none \\
LIDMark          & 6.0 & 5.9 & 68 & 93 & none \\
\midrule
Ours, relation $C_B$ & 66.7$^\ddagger$ & 5.9$^\ddagger$ & 0 & 0 & in-memory arrays \\
Ours, record model & 66.7$^\ddagger$ & 149.9$^\ddagger$ & 0.04 & 0 & $356$\,KiB NPZ \\
\bottomrule
\end{tabular*}
\end{table}

\section{Verification Limits and Record Privacy}
\label{app:limits}
\label{app:not_claimed}
This section examines how regeneration and targeted attacks affect verification, then discusses operating conditions and the protection of private records.

\subsection{Regeneration and Adaptive Attacks}
\label{app:stress_summary}

An attacker may seek to make a manipulated face pass verification or cause a genuine release to be rejected. To explore these deployment risks, we examine two routes: diffusion regeneration that may erase release-specific evidence, and targeted perturbations that directly alter verification scores. These tests probe whether the evidence remains useful when processing is deliberately chosen to undermine verification.

\noindent\textbf{Experimental settings.}
We evaluate $200$ FF++ test pairs in Table~\ref{tab:boundary_adaptive}. SD-Turbo regenerates both the release and its re-rendering at strengths $0.2$, $0.3$, and $0.4$, either over the whole crop or only the face region, without accessing the record. Strength controls how much noise is added before denoising: larger values allow more extensive regeneration and greater departure from the input. These values are diffusion settings, not percentages of pixels changed. Targeted attacks use an $L_\infty$ budget of $\epsilon=8/255$ on images in $[0,1]$. PGD~\citep{Madry2018PGD} takes $40$ steps of $\epsilon/8$ against the listed score; NES~\citep{Ilyas2018NES} uses $24$ score queries followed by one signed step. Transfer attacks use a DFDC-trained detector. Because these attacks lower suspiciousness for swaps and raise it for genuine releases, the resulting AUC measures ranking under this two-sided attack, rather than swap-evasion success at a fixed threshold.

\noindent\textbf{Scores and attack targets.}
The three columns measure template mismatch $1-C_B$, image-only score $S_{\rm df}$, and record-model score $S_{\rm fake}$. Record-model AUCs average three fixed FF++ checkpoints, using pre-sigmoid logits to avoid numerical saturation. Attacks optimize template or image-only scores; the record-model column measures attack transfer to a learned head that is not directly targeted. AUC $0.500$ denotes chance and $0.000$ an inverted ranking.

\begin{table}[H]
\centering\small\setlength{\tabcolsep}{3pt}\renewcommand{\arraystretch}{1.18}
\caption{\textbf{Detection under regeneration and targeted attacks} (AUC $\uparrow$). Stronger regeneration reduces separation, while template-targeted PGD also degrades the learned record model. Ranges cover three regeneration strengths, from weaker ($0.2$) to stronger ($0.4$).}
\label{tab:boundary_adaptive}
\begin{tabular*}{\linewidth}{@{\extracolsep{\fill}}ll ccc@{}}
\toprule
& & \multicolumn{3}{c}{Face-swap detection AUC $\uparrow$} \\
\cmidrule(lr){3-5}
Method & Setting & \shortstack{Template mismatch\\$1-C_B$} & \shortstack{Image-only\\$S_{\rm df}$} & \shortstack{Record model\\$S_{\rm fake}$} \\
\midrule
No attack & Original test queries & 1.000 & 0.996 & 0.999 \\
\addlinespace[4pt]
\multicolumn{5}{@{}l}{\textbf{Diffusion regeneration: SD-Turbo}}\\
SD-Turbo, $0.2$--$0.4$ & Whole crop & 0.665--0.893 & 0.835--0.973 & 0.725--0.893 \\
SD-Turbo, $0.2$--$0.4$ & Face region only & 0.783--0.954 & 0.873--0.988 & 0.771--0.947 \\
\addlinespace[4pt]
\multicolumn{5}{@{}l}{\textbf{Adaptive attacks ($\epsilon=8/255$)}}\\
NES (24 queries) & Score-only access & 0.971 & 0.992 & 0.922 \\
Transfer attack & DFDC-trained detector & 0.969 & 0.000 & 0.936 \\
White-box PGD & Image-only $S_{\rm df}$ & 0.964 & 0.000 & 0.926 \\
White-box PGD & Template $1-C_B$ & 0.000 & 0.994 & 0.081 \\
White-box PGD & $1-C_B$ and $S_{\rm df}$ & 0.124 & 0.000 & 0.797 \\
\bottomrule
\end{tabular*}
\end{table}

\noindent\textbf{Findings.}
The verification components provide complementary evidence but share vulnerabilities. Attacks against the image-only detector leave spectral comparison informative, whereas attacks against the template can also undermine the record model because it uses related spectral features. Adding a learned head therefore does not by itself protect the system from manipulation of its underlying evidence. Stronger diffusion regeneration also weakens separation, indicating that substantially regenerated copies require separate channel evaluation before automated verification. For deployment, these findings motivate testing attacks across verification components. Accordingly, processing identified as unsupported should yield \emph{unavailable}. Successful attacks, however, can satisfy the verification criteria without triggering human review. The component AUCs identify failure modes; evaluating their effect on final verdicts additionally requires testing the complete relation--key policy.

\noindent\textbf{Future directions.}
These results motivate exploring evidence distributed across frequency bands and spatial scales, with the aim of preserving discrimination under regeneration and delivery loss. Candidate designs should be evaluated against attacks targeting the complete relation--key policy, measuring both benign acceptance and manipulation rejection.

\subsection{Phase-Ring Transplantation}
\label{app:ring_transplantation}
An attacker holding the released image can copy its ring phase into a re-rendered face while retaining the latter's magnitudes. Phase is copied from the public release rather than recovered from a private template or key. However, the evaluated attack is supplied with the enrolled active-bin mask; its access assumption therefore includes support information from the record. Table~\ref{tab:ring_transplant} evaluates its effect on template-based and image-only detection.

\noindent\textbf{Evaluation settings.}
We test $200$ pairs per dataset, with thresholds calibrated for $1\%$ FPR on that dataset's benign training queries. AUC compares swapped queries with paired benign releases; TPR and FPR use the original thresholds, with the same benign baseline before and after transplantation. The calibration target does not guarantee a $1\%$ test FPR. The image-only detector uses fixed FF++ weights. ``Either flag'' means that the template or image-only test flags the query; this binary OR rule has no continuous score for AUC and is separate from the relation--key decision policy.

\begin{table}[H]
\centering\small\setlength{\tabcolsep}{2pt}\renewcommand{\arraystretch}{1.18}
\caption{\textbf{Detection before and after phase-ring transplantation.} Transplantation reduces template-score separation and detection at the original thresholds. FF++ retains substantial ranking ability despite low TPR, whereas DFDC approaches chance. TPR and FPR are percentages.}
\label{tab:ring_transplant}
\begin{tabular*}{\linewidth}{@{\extracolsep{\fill}}l ccc ccc cc@{}}
\toprule
& \multicolumn{3}{c}{Template mismatch $1-C_B$} & \multicolumn{3}{c}{Image-only $S_{\rm df}$} & \multicolumn{2}{c}{Either flag (OR)} \\
\cmidrule(lr){2-4}\cmidrule(lr){5-7}\cmidrule(lr){8-9}
Dataset & AUC $\uparrow$ & TPR $\uparrow$ & FPR $\downarrow$ & AUC $\uparrow$ & TPR $\uparrow$ & FPR $\downarrow$ & TPR $\uparrow$ & FPR $\downarrow$ \\
\midrule
\multicolumn{9}{@{}l}{\textbf{Before transplantation}}\\
FF++ & 1.000 & 99.5 & 2.0 & 0.996 & 99.5 & 14.0 & 99.5 & 15.5 \\
Celeb-DF & 0.990 & 95.0 & 1.5 & 0.687 & 18.0 & 4.5 & 95.0 & 6.0 \\
DFDC & 0.985 & 98.5 & 13.0 & 0.751 & 34.5 & 7.0 & 99.0 & 14.5 \\
\addlinespace[4pt]
\multicolumn{9}{@{}l}{\textbf{After transplantation}}\\
FF++ & 0.892 & 6.0 & 2.0 & 0.995 & 99.5 & 14.0 & 99.5 & 15.5 \\
Celeb-DF & 0.717 & 16.5 & 1.5 & 0.613 & 14.5 & 4.5 & 26.5 & 6.0 \\
DFDC & 0.579 & 15.0 & 13.0 & 0.724 & 30.0 & 7.0 & 32.0 & 14.5 \\
\bottomrule
\end{tabular*}
\end{table}
\noindent\textbf{Findings.}
Copied release phase can satisfy spectral criteria despite changed facial appearance. The learned joint policy confirms $84\%$ of attacked FF++ queries in the three-seed mean-score evaluation, so key evidence does not prevent this bypass. Image-only detection helps but weakens across datasets; stored-amplitude comparison does not restore separation. Bypasses can receive a \emph{confirmed} verdict and escape human review.

\par\medskip\Needspace{7\baselineskip}\noindent\textbf{Future directions.}
To resist evidence transplantation, a complementary direction is to couple the release signal to local facial content and verify consistency across facial regions. Evaluation should measure whether these dependencies improve rejection of transplanted evidence while preserving acceptance under supported processing; their effectiveness remains to be established.

\subsection{Verification Conditions and Interpretation}
Channel evaluation must assess both benign acceptance and manipulation rejection: processing can erase discriminative evidence while releases still pass. Failed comparisons identify discrepancies, not causes; unsupported processing or unresolved geometry requires an unavailable outcome.

Evaluation uses individual face crops with known enrollment geometry; scores do not localize facial edits. End-to-end multi-face detection and association, public signed receipts, and C2PA integration remain unevaluated. Human-review evaluation is not reported due to confidentiality obligations. The optional image-only review flag supplies separate evidence under Eq.~\eqref{eq:accept}.
\label{app:deployment_interpretation}
\label{app:boundary_reporting}

\par\medskip\noindent\textbf{Future directions.}
Deployment studies could combine channel-support and face-association confidence to identify when to return \emph{unavailable}. End-to-end evaluation should measure verification errors, automated coverage, and reviewer accuracy and workload, including attacks that are confirmed and escape referral.

\subsection{Privacy of Spectral Records}
\label{app:metadata_privacy}
Phase-only reconstruction reveals coarse facial structure (Fig.~\ref{fig:template}), making private records sensitive image-derived data. Access to templates, keys, and scores can therefore expose information beyond the reported verification outcome.

\par\medskip\noindent\textbf{Future directions.}
Future work could reduce recoverable facial information while retaining relation and key evidence, jointly evaluating verification performance, reconstruction, and identity leakage. Reporting protocols should authenticate outcomes with minimal disclosure and assess leakage through repeated verification requests.

\newpage

\end{document}